\documentclass[sigconf,dvipsnames,nonacm]{acmart}

\graphicspath{{figures/}}
\DeclareGraphicsExtensions{.pdf}

\usepackage[normalem]{ulem}

\usepackage{stmaryrd}
\usepackage{xspace}
\usepackage{amsmath}
\usepackage{amsthm}
\usepackage{physics}
\usepackage{stmaryrd} 
\usepackage[old]{old-arrows} 

\makeatletter
\providecommand{\bigsqcap}{%
  \mathop{%
    \mathpalette\@updown\bigsqcup
  }%
}
\newcommand*{\@updown}[2]{%
  \rotatebox[origin=c]{180}{$\m@th#1#2$}%
}
\makeatother

\usepackage{float}

\usepackage{balance}

\usepackage{subfig}
\usepackage[font=footnotesize]{caption}
\usepackage{graphicx}
\usepackage[export]{adjustbox}
\usepackage{wrapfig}

\usepackage{algorithm}
\usepackage{algpseudocode}
\algrenewcommand\algorithmicindent{0.6em}

\usepackage{listings}
\usepackage{lipsum}

\usepackage{textgreek}

\usepackage{enumitem}

\usepackage{booktabs}
\usepackage{array}

\usepackage{booktabs}
\usepackage{multirow}

\usepackage{nccmath}

\usepackage[most]{tcolorbox}
\newcommand*\circled[1]{\tikz[baseline=(char.base)]{
    \node [shape=circle,draw=black,inner sep=1pt,fill=white,
           line width=0.7pt,text=black,minimum size=3.3mm] (char) {\footnotesize\textsf{#1}};
}}

\newtheorem{theorem}{Theorem}

\newtheorem{definition}{Definition}

\newtheorem{property}{Property}
\newtheorem{question}{Question}

\newcommand{\omitsp}[1]{}

\newif\ifhidechanges

\newif\ifhideremovals

\usepackage{totcount}
\usepackage{soul}
\newif\ifhidecomments

\newtotcounter{commentcounter}
\DeclareDocumentCommand \rutvik { o m } {%
    \stepcounter{commentcounter}%
    \ifhidecomments
        \ignorespaces
    \else \IfNoValueTF{#1}{%
        \textsf{\textcolor{Cyan}{\footnotesize[\textbf{Rutvik:}\;#2]}}%
    }{%
        \textsf{\textcolor{lightgray}{\footnotesize[\textbf{Rutvik:}\;#2]}}%
    }%
    \fi
}
\DeclareDocumentCommand \chris { o m } {%
    \stepcounter{commentcounter}%
    \ifhidecomments
        \ignorespaces
    \else \IfNoValueTF{#1}{%
        \textsf{\textcolor{Red}{\footnotesize[\textbf{Chris:}\;#2]}}%
    }{%
        \textsf{\textcolor{lightgray}{\footnotesize[\textbf{Chris:}\;#2]}}%
    }%
    \fi
}
\DeclareDocumentCommand \adam { o m } {%
    \stepcounter{commentcounter}%
    \ifhidecomments
        \ignorespaces
    \else \IfNoValueTF{#1}{%
        \textsf{\textcolor{Green}{\footnotesize[\textbf{Adam:}\;#2]}}%
    }{%
        \textsf{\textcolor{lightgray}{\footnotesize[\textbf{Adam:}\;#2]}}%
    }\fi
}
\DeclareDocumentCommand \neil { o m } {%
    \stepcounter{commentcounter}%
    \ifhidecomments
        \ignorespaces
    \else \IfNoValueTF{#1}{%
        \textsf{\textcolor{Magenta}{\footnotesize[\textbf{Neil:}\;#2]}}%
    }{%
        \textsf{\textcolor{lightgray}{\footnotesize[\textbf{Neil:}\;#2]}}%
    }%
    \fi
}
\DeclareDocumentCommand \alan { o m } {%
    \stepcounter{commentcounter}%
    \ifhidecomments
        \ignorespaces
    \else \IfNoValueTF{#1}{%
        \textsf{\textcolor{Orange}{\footnotesize[\textbf{Alan:}\;#2]}}%
    }{%
        \textsf{\textcolor{lightgray}{\footnotesize[\textbf{Alan:}\;#2]}}%
    }%
    \fi
}

\definecolor{ClearSky}{RGB}{173, 237, 255}
\definecolor{Mist}{RGB}{232, 250, 255}
\definecolor{Cereal}{RGB}{255, 234, 150}
\definecolor{Wheat}{RGB}{255, 249, 223}
\definecolor{OceanDepth}{RGB}{26, 31, 161}
\definecolor{Abyss}{RGB}{15, 19, 96}
\definecolor{Spearmint}{RGB}{149, 251, 190}
\definecolor{Lichen}{RGB}{232, 253, 240}
\definecolor{Eraser}{RGB}{255, 195, 254}
\definecolor{RoseOfSharon}{RGB}{255, 233, 255}

\definecolor{HotPink}{RGB}{247, 0, 103}
\definecolor{HotBlue}{RGB}{13, 101, 255}
\definecolor{HotGreen}{RGB}{19, 180, 10}
\definecolor{HotOrange}{RGB}{255, 87, 51}
\definecolor{LightGray}{gray}{0.65}
\definecolor{MediumGray}{gray}{0.45}

\newcommand{\toolnameplain}{Declassiflow\xspace}
\newcommand{\toolname}{\textsc{\toolnameplain}\xspace}

\newcommand{\decl}{\text{\scriptsize Decl}}

\newcommand{\programraw}{P}
\DeclareDocumentCommand \program { o } {%
    \IfNoValueTF{#1}{%
        \programraw
    }{%
        \programraw_{#1}
    }
}
\newcommand{\policyraw}{\Pi}
\DeclareDocumentCommand \policy { o } {%
    \IfNoValueTF{#1}{%
        \policyraw
    }{%
        \policyraw_{#1}
    }
}

\DeclareDocumentCommand \varset { o } {%
    \IfNoValueTF{#1}{%
        V
    }{%
        V_{#1}
    }
}
\newcommand{\allvars}{\mathbb{V}}
\newcommand{\varlatt}{2^\allvars}

\newcommand{\allblocks}{\mathbb{B}}
\newcommand{\blockraw}{B}
\DeclareDocumentCommand \block { o o } {%
    \IfNoValueTF{#1}{%
        \blockraw
    }{%
        \IfNoValueTF{#2}{%
            \blockraw_{#1}
        }{%
            \blockraw_{#1}^{(#2)}
        }
    }
}
\newcommand{\blocksetraw}{\mathbf{\blockraw}}
\DeclareDocumentCommand \blockset { o } {%
    \IfNoValueTF{#1}{%
        \blocksetraw
    }{%
        \blocksetraw(#1)
    }
}

\newcommand{\regionraw}{R}
\DeclareDocumentCommand \region { o } {%
    \IfNoValueTF{#1}{%
        \regionraw
    }{%
        \regionraw_{#1}
    }
}

\newcommand{\alledges}{\mathbb{E}}
\newcommand{\edgeraw}{e}
\DeclareDocumentCommand \edge { o o } {%
    \IfNoValueTF{#1}{%
        \edgeraw
    }{%
        \IfNoValueTF{#2}{%
            \edgeraw_{#1}
        }{%
            \edgeraw_{#1}^{#2}
        }
    }
}
\DeclareDocumentCommand \inedge { o } {%
    \IfNoValueTF{#1}{%
        \edge[][\textup{in}]
    }{%
        \edge[#1][\textup{in}]
    }
}
\DeclareDocumentCommand \outedge { o } {%
    \IfNoValueTF{#1}{%
        \edge[][\textup{out}]
    }{%
        \edge[#1][\textup{out}]
    }
}

\DeclareDocumentCommand \inst { o } {%
    \IfNoValueTF {#1} {%
        I
    }{%
        I_{#1}
    }%
}
\DeclareDocumentCommand \codeinst { o } {%
    \IfNoValueTF {#1} {%
        \texttt{I}
    }{%
        \texttt{I}$_\text{#1}$
    }%
}

\newcommand{\edgesetraw}{E}
\newcommand{\inedgesetraw}{\edgesetraw_\textup{in}}
\DeclareDocumentCommand \inedgeset { o o } {%
    \IfNoValueTF {#1} {%
        \IfNoValueTF {#2} {%
            \inedgesetraw(\block)
        }{}
    }{%
        \IfNoValueTF {#2} {%
            \inedgesetraw(\block_{#1})
        }{%
            \inedgesetraw({#1})
        }%
    }%
}
\newcommand{\outedgesetraw}{\edgesetraw_\textup{out}}
\DeclareDocumentCommand \outedgeset { o o } {%
    \IfNoValueTF {#1} {%
        \IfNoValueTF {#2} {%
            \outedgesetraw(\block)
        }{}
    }{%
        \IfNoValueTF {#2} {%
            \outedgesetraw(\block_{#1})
        }{%
            \outedgesetraw({#1})
        }%
    }%
}

\newcommand{\dfvsymbol}{{\widehat{\knownssymbol}_\alledges}}
\DeclareDocumentCommand \alldfvs { o } {%
    \IfNoValueTF {#1} {%
        \dfvsymbol
    }{%
        \dfvsymbol_{#1}
    }%
}
\DeclareDocumentCommand \dfv { o o } {%
    \IfNoValueTF {#1} {%
        \IfNoValueTF {#2} {%
            {\alldfvs(\edge)}
        }{}
    }{%
        \IfNoValueTF {#2} {%
            {\alldfvs(\edge_{#1})}
        }{%
            {\alldfvs(#1)}
        }%
    }%
}

\DeclareDocumentCommand \trace { o } {%
    \IfNoValueTF{#1}{%
        t
    }{%
        t_{#1}
    }
}
\DeclareDocumentCommand \traceset { o } {%
    \IfNoValueTF{#1}{%
        T
    }{%
        T_{#1}
    }
}
\newcommand{\alltraces}{\mathbb{T}}

\newcommand{\allpaths}{\mathbb{P}}
\DeclareDocumentCommand \path { o } {%
    \IfNoValueTF{#1}{%
        p
    }{%
        p_{#1}
    }
}

\newcommand{\knownssymbol}{K}
\newcommand{\alledgeknowns}{{\knownssymbol_{\alledges}}}
\DeclareDocumentCommand \edgeknowns { o o } {%
    \IfNoValueTF {#1} {%
        \IfNoValueTF {#2} {%
            \alledgeknowns(\edge)
        }{}
    }{%
        \IfNoValueTF {#2} {%
            \alledgeknowns(\edge[#1])
        }{%
            \alledgeknowns({#1}_{#2})
        }%
    }%
}

\newcommand{\phiinsts}[1][]{\text{phi}\qty(\block[#1])}

\newcommand{\instoutput}[1][]{\text{out}(\inst[#1])}
\newcommand{\instinputs}[1][]{\text{in}(\inst[#1])}

\newcommand{\entrynode}{\textsc{Entry}}
\newcommand{\exitnode}{\textsc{Exit}}

\newcommand{\define}[1]{\emph{#1}}

\newcommand{\dom}{\,\text{dom}\,}

\newcommand{\pdom}{\,\text{pdom}\,}

\newcommand{\inv}{^{-1}}

\algnewcommand{\LineComment}[1]{\State \(\triangleright\) #1}

\let\emptyset\varnothing

\DeclareDocumentCommand \supernode { o } {%
    \IfNoValueTF{#1}{%
        \vb{\block}
    }{%
        \vb{\block}^{(#1)}
    }
}

\newcommand{\ith}{$i$-th\xspace}

\newcommand{\pink}[1]{\text{\color{HotPink}#1}}
\newcommand{\blue}[1]{\text{\color{HotBlue}#1}}

\newcommand{\pinktt}[1]{\texttt{\color{HotPink}#1}}
\newcommand{\bluett}[1]{\texttt{\color{HotBlue}#1}}
\newcommand{\greentt}[1]{\texttt{\color{HotGreen}#1}}

\newcommand{\pinkm}[1]{{\color{HotPink}#1}}

\newcommand{\writtenphi}{\textup{\straightphi}}
\newcommand{\nonphi}{non-\writtenphi\xspace}
\newcommand{\Nonphi}{Non-\writtenphi\xspace}
\newcommand{\phifunction}{\writtenphi-function\xspace}
\newcommand{\phifunctions}{\writtenphi-functions\xspace}
\newcommand{\inductivephi}{inductive-\writtenphi\xspace}

\newcommand{\inductivephis}{inductive-\writtenphi's\xspace}

\DeclareDocumentCommand \fwdtransfraw { o } {%
    \IfNoValueTF {#1} {%
            \varoverrightarrow{f\,\,}
        }{%
            \varoverrightarrow{f\,\,}_{\!\!\!#1}
        }%
}
\DeclareDocumentCommand \fwdtransf { o o o } {%
    \IfNoValueTF{#1}{%
        \IfNoValueTF{#2}{%
            \IfNoValueTF{#3}{%
                \fwdtransfraw[i\alpha]
            }{}
        }{}
    }{%
        \IfNoValueTF{#2}{%
            \IfNoValueTF{#3}{%
                \fwdtransfraw[i\alpha]\qty(#1)
            }{}
        }{
            \IfNoValueTF{#3}{%
                \fwdtransfraw[#1 #2]\qty(\dfv)
            }{
                \fwdtransfraw[#1 #2]\qty(#3)
            }
        }
    }
}
\DeclareDocumentCommand \bwdtransfraw { o } {%
    \IfNoValueTF {#1} {%
        \varoverleftarrow{f}
    }{%
        \varoverleftarrow{f}_{\!\!#1}
    }%
}
\DeclareDocumentCommand \bwdtransf { o o o } {%
    \IfNoValueTF{#1}{%
        \IfNoValueTF{#2}{%
            \IfNoValueTF{#3}{%
                \bwdtransfraw[j\beta]
            }{}
        }{}
    }{%
        \IfNoValueTF{#2}{%
            \IfNoValueTF{#3}{%
                \bwdtransfraw[j\beta]\qty(#1)
            }{}
        }{
            \IfNoValueTF{#3}{%
                \bwdtransfraw[#1 #2]\qty(\dfv)
            }{
                \bwdtransfraw[#1 #2]\qty(#3)
            }
        }
    }
}

\newcommand{\transmitsraw}{\textsc{Tx}}
\DeclareDocumentCommand \transmits { o } {%
    \IfNoValueTF{#1}{%
        \transmitsraw_i
    }{%
        \transmitsraw_{#1}
    }
}
\DeclareDocumentCommand \transmitvars { o } {%
    \transmitsraw\qty(\allvars)
}
\DeclareDocumentCommand \transmitblocks {} {%
    \transmitsraw(\allblocks)
}

\DeclareDocumentCommand \solvs { o o } {%
    \IfNoValueTF{#1}{%
        \IfNoValueTF{#2}{%
            \textsc{Solv}_{i}\qty(C)
        }{}
    }{%
        \IfNoValueTF{#2}{%
            \textsc{Solv}_{i}\qty(#1)
        }{%
            \textsc{Solv}_{#1}\qty(#2)
        }
    }
}
\DeclareDocumentCommand \fwdsolvs { o o } {%
    \IfNoValueTF{#1}{%
        \IfNoValueTF{#2}{%
            \textsc{FSolv}\qty(\inst, C)
        }{}
    }{%
        \IfNoValueTF{#2}{%
            \textsc{FSolv}\qty(\inst, #1)
        }{%
            \textsc{FSolv}\qty(#1, #2)
        }
    }
}
\newcommand{\bwdsolvsraw}{{\textsc{BSolv}}}
\DeclareDocumentCommand \bwdsolvs { o o } {%
    \IfNoValueTF{#1}{%
        \IfNoValueTF{#2}{%
            \bwdsolvsraw\qty(\inst, C)
        }{}
    }{%
        \IfNoValueTF{#2}{%
            \bwdsolvsraw\qty(\inst, #1)
        }{%
            \bwdsolvsraw\qty(#1, #2)
        }
    }
}

\newcommand{\phisolvsraw}{{\textsc{PhiSolv}}}
\newcommand{\phisolvsfull}[3]{\phisolvsraw_{#1 #2}\qty(#3)}
\DeclareDocumentCommand \phisolvs { o o o } {%
    \IfNoValueTF{#1}{%
        \IfNoValueTF{#2}{%
            \IfNoValueTF{#3}{%
                \phisolvsfull{i}{\beta}{\dfv}
            }{}
        }{}
    }{%
        \IfNoValueTF{#2}{%
            \IfNoValueTF{#3}{%
                \phisolvsfull{i}{\beta}{#1}
            }{}
        }{
            \IfNoValueTF{#3}{%
                \phisolvsfull{#1}{#2}{\dfv}
            }{
                \phisolvsfull{#1}{#2}{#3}
            }
        }
    }
}
\newcommand{\phifsolvsfull}[3]{\textsc{PhiFSolv}_{#1 #2}\qty(#3)}
\DeclareDocumentCommand \phifsolvs { o o o } {%
    \IfNoValueTF{#1}{%
        \IfNoValueTF{#2}{%
            \IfNoValueTF{#3}{%
                \phifsolvsfull{i}{\alpha}{\dfv}
            }{}
        }{}
    }{%
        \IfNoValueTF{#2}{%
            \IfNoValueTF{#3}{%
                \phifsolvsfull{i}{\alpha}{#1}
            }{}
        }{
            \IfNoValueTF{#3}{%
                \phifsolvsfull{#1}{#2}{\dfv}
            }{
                \phifsolvsfull{#1}{#2}{#3}
            }
        }
    }
}
\newcommand{\phibsolvsfull}[3]{\textsc{PhiBSolv}_{#1 #2}\qty(#3)}
\DeclareDocumentCommand \phibsolvs { o o o } {%
    \IfNoValueTF{#1}{%
        \IfNoValueTF{#2}{%
            \IfNoValueTF{#3}{%
                \phibsolvsfull{j}{\beta}{\dfv}
            }{}
        }{}
    }{%
        \IfNoValueTF{#2}{%
            \IfNoValueTF{#3}{%
                \phibsolvsfull{j}{\beta}{#1}
            }{}
        }{
            \IfNoValueTF{#3}{%
                \phibsolvsfull{#1}{#2}{\dfv}
            }{
                \phibsolvsfull{#1}{#2}{#3}
            }
        }
    }
}

\DeclareDocumentCommand \defns { o } {%
    \IfNoValueTF{#1}{%
        \textsc{Defn}_i
    }{%
        \textsc{Defn}_{#1}
    }
}

\DeclareDocumentCommand \diredge { o o } {%
    \IfNoValueTF{#1}{%
        \IfNoValueTF{#2}{%
            i {\shortrightarrow} j
        }{}
    }{%
        \IfNoValueTF{#2}{%
            i {\shortrightarrow} #1
        }{%
            #1 {\shortrightarrow} #2
        }
    }
}

\DeclareDocumentCommand \preds { o } {%
    \IfNoValueTF{#1}{%
        \text{pred}\qty(\block[i])
    }{%
        \text{pred}\qty(#1)
    }
}

\DeclareDocumentCommand \succs { o } {%
    \IfNoValueTF{#1}{%
        \text{succ}\qty(\block[i])
    }{%
        \text{succ}\qty(#1)
    }
}

\DeclareDocumentCommand \transmitinsts { o } {%
    \IfNoValueTF{#1}{%
        \text{transmits}\qty(\block)
    }{%
        \text{transmits}\qty(\block[#1])
    }
}

\DeclareDocumentCommand \phiinsts { o } {%
    \IfNoValueTF{#1}{%
        \text{phi}\qty(\block)
    }{%
        \text{phi}\qty(\block[#1])
    }
}

\DeclareDocumentCommand \attached { o } {%
    \IfNoValueTF{#1}{%
        \text{att}\qty(\edge)
    }{%
        \text{att}\qty(\edge[#1])
    }
}

\DeclareDocumentCommand \codevar { m o o } {%
    \IfNoValueTF{#2}{%
        \texttt{#1}%
    }{%
        \IfNoValueTF{#3}{%
            \texttt{#1$_\text{#2}$}%
        }{%
            \texttt{#1$_\text{#2}^\text{(#3)}$}%
        }%
    }%
}
\DeclareDocumentCommand \codevarprime { m o o } {%
    \IfNoValueTF{#2}{%
        \texttt{$\text{#1}'$}%
    }{%
        \IfNoValueTF{#3}{%
            \texttt{$\text{#1}'_\text{#2}$}%
        }{%
            \texttt{$\text{#1}'_\text{#2}^\text{(#3)}$}%
        }%
    }%
}

\DeclareDocumentCommand \var { m o o } {%
    \IfNoValueTF{#2}{%
        #1
    }{%
        \IfNoValueTF{#3}{%
            #1_{#2}
        }{%
            #1_{#2}^{(#3)}
        }%
    }%
}

\newcommand{\too}{\allowbreak{\shortrightarrow}\allowbreak}

\DeclareDocumentCommand \protecc { o } {%
    \IfNoValueTF{#1}{%
        \texttt{protect}\xspace%
    }{%
        \ifthenelse{\equal{#1}{}}{%
            \texttt{protect($\cdot$)}\xspace%
        }{%
            \texttt{protect(#1)}\xspace%
        }%
    }%
}

\newcommand{\allblockknowns}{{{\widehat\knownssymbol}_\allblocks}}
\DeclareDocumentCommand \blockknowns { o o } {%
    \IfNoValueTF {#1} {%
        \IfNoValueTF {#2} {%
            \allblockknowns\qty(\block)
        }{}
    }{%
        \IfNoValueTF {#2} {%
            \allblockknowns\qty(\block[#1])
        }{%
            \allblockknowns\qty({#1}_{#2})
        }%
    }%
}

\newcommand{\allfrontiers}{\mathcal{F}}
\DeclareDocumentCommand \frontierraw { o } {%
    \IfNoValueTF {#1} {%
        \allfrontiers
    }{%
        \allfrontiers_{#1}
    }%
}
\DeclareDocumentCommand \frontier { m } {%
    \ifthenelse{\equal{#1}{}}{%
        \allfrontiers($\cdot$)
    }{%
        \allfrontiers(#1)
    }%
}
\DeclareDocumentCommand \frontierinv { m } {%
    \ifthenelse{\equal{#1}{}}{%
        \allfrontiers\inv($\cdot$)
    }{%
        \allfrontiers\inv(#1)
    }%
}

\newcommand{\klee}{\textsc{Klee}\xspace}

\DeclareDocumentCommand \assert { o } {%
    \IfNoValueTF {#1} {%
        \texttt{assert($\dots$)}\xspace
    }{%
        \texttt{assert(#1)}\xspace
    }%
}

\DeclareDocumentCommand \transmit { o } {%
    \IfNoValueTF {#1} {%
        \texttt{Tr($\dots$)}\xspace
    }{%
        \texttt{Tr(#1)}\xspace
    }%
}

\newcommand{\specbarr}{\texttt{SPEC\_BARR}\xspace}

\DeclareDocumentCommand \func { m o o } {%
    \IfNoValueTF {#2} {%
        \IfNoValueTF {#3} {%
            \texttt{#1($\dots$)}\xspace
        }{}%
    }{%
        \IfNoValueTF {#3} {%
            \texttt{$\text{#1}_\text{#2}$($\dots$)}\xspace
        }{%
            \texttt{$\text{#1}_\text{#2}$(#3)}\xspace
        }%
    }%
}
\DeclareDocumentCommand \funcprime { m o o } {%
    \IfNoValueTF {#2} {%
        \IfNoValueTF {#3} {%
            \texttt{#1$'$($\dots$)}\xspace
        }{}%
    }{%
        \IfNoValueTF {#3} {%
            \texttt{$\text{#1}_\text{#2}'$($\dots$)}\xspace
        }{%
            \texttt{$\text{#1}_\text{#2}'$(#3)}\xspace
        }%
    }%
}

\DeclareDocumentCommand \set { m } {%
    \{{#1}\}
}

\begin{abstract}

Speculative execution attacks undermine the security of constant-time programming, the standard technique used
to prevent microarchitectural side channels in security-sensitive software such as cryptographic code.
Constant-time code must therefore also deploy a defense against speculative execution attacks to prevent
leakage of secret data stored in memory or the processor registers. Unfortunately, contemporary defenses, such
as speculative load hardening (SLH), can only satisfy this strong security guarantee at a very high performance
cost.

This paper proposes \toolname, a static program analysis and protection framework to \emph{efficiently} protect
constant-time code from speculative leakage. \toolname models ``attacker knowledge''---data which is inherently
transmitted (or, \emph{implicitly} declassified) by the code's \emph{non-speculative} execution---and
statically removes protection on such data from points in the program where it is already guaranteed to leak
non-speculatively. Overall, \toolname ensures that data which never leaks during the non-speculative execution
does not leak during speculative execution, but with lower overhead than conservative protections like SLH.

\end{abstract}

\begin{document}

\title[\toolname: A Static Analysis for Modeling Non-Spec. Knowledge to Relax Spec. Execution
Security Measures (Full Version)]{\toolname: A Static Analysis for Modeling Non-Speculative Knowledge to Relax Speculative Execution
Security Measures (Full Version)}

\author{Rutvik Choudhary}
\email{rutvikc2@illinois.edu}
\affiliation{%
  \institution{University of Illinois Urbana-Champaign}
  \city{Urbana}
  \state{Illinois}
  \country{USA}
}
\author{Alan Wang}
\email{alanlw2@illinois.edu}
\affiliation{%
  \institution{University of Illinois Urbana-Champaign}
  \city{Urbana}
  \state{Illinois}
  \country{USA}
}
\author{Zirui Neil Zhao}
\email{ziruiz6@illinois.edu}
\affiliation{%
  \institution{University of Illinois Urbana-Champaign}
  \city{Urbana}
  \state{Illinois}
  \country{USA}
}
\author{Adam Morrison}
\email{mad@cs.tau.ac.il}
\affiliation{%
  \institution{Tel Aviv University}
  \city{Tel Aviv}
  \country{Israel}
}
\author{Christopher W. Fletcher}
\email{cwfletch@illinois.edu}
\affiliation{%
  \institution{University of Illinois Urbana-Champaign}
  \city{Urbana}
  \state{Illinois}
  \country{USA}
}
\maketitle

\sloppy

\thispagestyle{empty}

\section{Introduction}
\label{sec:intro}

Security-sensitive programs, such as cryptographic software, perform computations over secret data (e.g., cipher
keys and plaintext or personal information). Secure software must prevent its secrets from being ``leaked'' over
microarchitectural side channels, which occur when secret data is passed as the operand 
to a \emph{transmitter} instruction. A transmitter is an instruction whose execution creates operand-dependent
hardware resource patterns that can potentially be observed (``received'') by the attacker, allowing the
attacker to learn information about the transmitter's operand. Classic examples of transmitters are load and
branch instructions, whose execution makes operand-dependent changes to the cache
state~\cite{flush+,last_level_cache_practical} and instruction sequence. However, numerous other ``variable
time'' instructions are also considered transmitters~\cite{FPU_leaky,Mult_leaky,practical_doprogramming}.

Traditionally, the standard technique for preventing secret leakage over microarchitectural side channels is to
use \emph{constant-time programming} (also called \emph{data-oblivious programming}). Constant-time code
performs its computation without passing secret-dependent data as arguments to transmitter
instructions~\cite{constanttimersa,curve25519,poly1305,pc_model,Raccoon,FPU_leaky}.

Unfortunately, the discovery of speculative execution (or Spectre)
attacks~\cite{spectre,ret2spec,smother,spectre-BHI,retbleed} undermines the constant-time
approach~\cite{smother,oisa,deianpldi,SpectreDeclassified}. The problem is that constant-time guarantees are
based on correct execution semantics and may not hold in an illegal mis-speculated execution created by a
speculative execution attack. For example, misprediction of a loop branch~\cite{oisa}, function
return~\cite{deianpldi}, indirect call~\cite{smother}, and so on can cause the processor to jump to a
transmitter with the transmitter's operand holding secret data, even though the transmitter's operand would
never contain secret data in a correct execution (see Figure~\ref{fig:intro}).

\begin{figure}[t]
\begin{lstlisting}[]
  for (int i = 0; @i < NUM_ROUNDS;@ i++) {
    !S! = AES_Round(!S!, round_key[i]);
  }
  Tr(!S!);
\end{lstlisting}
\caption{
Example of constant-time code breaking due to speculative execution~\cite{oisa}. Misprediction of the loop branch (\bluett{i$\,$<$\,$NUM\_ROUNDS}) can cause
an intermediate value of the AES state \pinktt{S} to be passed to a transmitter, \texttt{Tr($\cdot$)}.\vspace{-1pt}}
\label{fig:intro}
\end{figure}

Consequently, constant-time code must additionally deploy a defense against speculative execution attacks.
Importantly, this defense must prevent speculative leakage not only of speculatively-accessed data (read from
memory under mis-speculation)~\cite{stt} but also of \emph{non-speculatively-accessed data} that already exists
in processor registers when mis-speculation begins (as in Figure~\ref{fig:intro}). Contemporary defenses can
only satisfy this strong security guarantee by blocking speculation of all transmitters, which incurs a high
performance cost~\cite{SSLH,USLH}. We therefore ask: \emph{how can constant-time code be efficiently protected
from speculative execution attacks?}

We answer this question with \toolname, a static program analysis and protection framework that can relax
speculative execution defenses. Our approach enforces the security property proposed by Speculative Privacy
Tracking (SPT)~\cite{spt}: data that never leaks during non-speculative execution does not leak during
speculative execution. This property implies that data which gets \emph{implicitly declassified}, due to being
passed as the operand of a transmitter in the program's non-speculative execution, does not need to be protected
during the program's speculative execution. This enables the safe removal of protection mechanisms and
commensurately lower performance overhead.

Leveraging the SPT security property to reap performance benefits is non-trivial, however. SPT is only able to
achieve gains by introducing hardware mechanisms for dynamically tracking non-speculative leakage and disabling
protection at run-time. But SPT hardware is not available in current processors and its future adoption status
is not clear. In this paper, we leverage the SPT security property \emph{purely in software}. The resulting
approach can be deployed to improve the performance of constant-time code today. \toolname can also identify
protection relaxations that SPT hardware misses, because \toolname is a static program analysis that reasons
about all possible program executions, whereas SPT hardware operates only based on the program's current
execution.

In a nutshell, \toolname performs program analysis to determine the ``attacker's knowledge'' at every edge in
the program's non-speculative control-flow graph, where ``attacker knowledge'' refers to the data guaranteed to
be implicitly declassified (declared non-secret) by the non-speculative execution if said control transfer
occurs at run-time. \toolname can thus identify data that is \emph{guaranteed} to leak if execution reaches a
program point (although it may not leak at that point, but only later in the execution).

We use \toolname's analysis to relax speculative execution protections, such as SLH, for several constant-time
programs. By reasoning about attacker knowledge, \toolname is able to deduce that many transmitters (e.g.,
loads) leak information about the ``same thing'' (e.g., the base address of an array) and that this information
is guaranteed to be known in the program's non-speculative execution. This enables \toolname to reduce overhead
significantly; in some cases, replacing all protection instrumentation with a single mechanism (e.g., a barrier)
that guarantees the program is entered non-speculatively.

To summarize, we make the following contributions.
\begin{enumerate}
\item We propose an abstraction, \emph{non-speculative knowledge}, for deducing a program region where a
variable will ``inevitably'' be leaked in the program's non-speculative execution.
\item We propose a novel program analysis that can calculate non-speculative attacker knowledge at each program
edge, and strategies for placing protection primitives based on said non-speculative attacker knowledge.
\item We evaluate the impact of our analysis on three constant-time benchmarks and demonstrate that our analysis
leads to more efficiently protected programs.
\end{enumerate}

Our analysis is open source, and can be found at \url{https://github.com/FPSG-UIUC/declassiflow}.
\section{Background}

\subsection{Programs, Executions, and Traces}
\label{sec:building-blocks}

\subsubsection{The Building Blocks of a Program}

We consider programs written in the LLVM assembly language~\cite{llvm-lang}, which uses static single assignment
(SSA) (Section~\ref{sec:SSA}).

A \define{program} is a list of instructions that perform computations on variables. $\allvars$ denotes the set
of all variables in the program, and $\varlatt$ denotes the powerset of $\allvars$.

Instructions in a program are partitioned into smaller lists, \define{basic blocks} (or \define{blocks} for
short), defined in the usual way. They are typically denoted as $\block$, often with a subscript. We use
$\allblocks$ to denote the set of all blocks in the program.

A \define{control-flow edge} (or \define{edge} for short) is an ordered pair $(\block[i], \block[j])$ for some
$\block[i], \block[j] \in \allblocks$. A control-flow edge exists between $\block[i]$ and $\block[j]$ if it is
theoretically possible to execute $\block[j]$ immediately after $\block[i]$ has been executed. Edges are denoted
as $\edge$, often with a subscript. We use $\alledges$ to denote the set of all edges in the program. For any
block $\block$, we denote the set of its input edges and output edges as $\inedgeset$ and $\outedgeset$
respectively.

A \define{control-flow graph} is a directed graph where the nodes are given by $\allblocks$ and the edges are
given by $\alledges$. We define a node $\entrynode\in\allblocks$ which is the singular point of entry for the
program as well as a node $\exitnode\in\allblocks$ which is the singular exit point of the program.%
\footnote{Not all paths through the control-flow graph will reach $\exitnode$, e.g. in the case of
non-terminating loops.} By definition, $\inedgeset[\entrynode][]=\emptyset$ and
$\outedgeset[\exitnode][]=\emptyset$.

For any two blocks $\block[i], \block[j] \in \allblocks$, $\block[i]$ \define{dominates} $\block[j]$, denoted
$\block[i] \dom \block[j]$, if every path from $\entrynode$ to $\block[j]$ goes through $\block[i]$. $\block[j]$
\define{post-dominates} $\block[i]$, denoted $\block[j] \pdom \block[i]$, if every path from $\block[i]$ to
$\exitnode$ goes through $\block[j]$. $\block[i]$ is a \emph{predecessor} of $\block[j]$ if ($\block[i]$,
$\block[j]$) $\in \alledges$.

We say that an edge $\edge = (\block[i], \block[j])$ dominates a block $\block'$ if $\block[j] \dom \block'$.
Similarly, $\block'$ dominates edge $\edge$ if $\block' \dom \block[i]$. A variable $x$ dominates edge $\edge$
if the block in which $x$ is defined, denoted as $\block[x]$, dominates $\block[i]$. Similarly, $\edge$
dominates $x$ if $\block[j] \dom \block[x]$.

A \define{region}, typically denoted $\region$, is a set of blocks such that: \circled{1} there is one block in
$\region$, known as the \define{header}, that dominates all others; \circled{2} for any two blocks $\block[i]
\in \allblocks$ and $\block[j] \in \region$, if there is a path from $\block[i]$ to $\block[j]$ that doesn't
contain the header, then $\block[i]$ is in $\region$~\cite{dragon-book}.

Edge $(\block[i],\block[j])$ is a \define{back edge} if $\block[j]\dom \block[i]$. By convention, every block
dominates itself, and so self edges (edges where $\block[i]=\block[j]$) are considered back edges. If there is a
back edge in the control-flow graph, then there is a cycle. Cycles in the control-flow graph are typically
created by using loop constructs (e.g. \texttt{for} and \texttt{while}).

\subsubsection{Executions and Traces}
\label{sec:exec-traces}

A \define{non-speculative execution} of a program $P$ on some input is the sequence of instructions $P$ executes
according to the semantics of the LLVM language~\cite{llvm-lang}. We say that the $k$-th instruction in an
execution occurs at \define{time} $k$. An execution \define{traverses} edge $\edge = (\block[i], \block[j])$ at
time $k$ if the $k$-th and $(k+1)$-th instructions are the last instruction in $\block[i]$ and the first
instruction in $\block[j]$, respectively.

A \define{trace}, typically denoted as $\trace$, is a sequence of edges. A trace $\trace$ is \define{realizable}
if, for some execution of the program, $\trace$ is the sequence of edges traversed by the execution. Realizable
traces thus model the control flow of executions. We refer to them interchangeably for brevity, understanding
that every realizable trace is associated with an execution. The set of all realizable traces of the program is
denoted $\alltraces$. An edge $\edge$ is \define{realizable} if $e \in \trace$ for some $\trace \in \alltraces$.

We use $\allpaths$ to denote all possible paths through the control-flow graph, including those
that do not correspond to a realizable trace. By definition, $\alltraces \subseteq \allpaths$.

We now extend our execution semantics to capture \emph{speculative executions}. We consider control-flow
speculation of branches whose speculative target is consistent with the control-flow graph. That is, for
$\block[i] \neq \block[j]$, a speculative execution executes instruction $\inst \in \block[i]$ followed by
instruction $\inst' \in \block[j]$ only if $\inst$ is the last instruction in $\block[i]$, $\inst'$ is the first
instruction in $\block[j]$, and $\edge = (\block[i], \block[j]) \in \alledges$. The semantics of the LLVM
assembly language~\cite{llvm-lang} can be extended to model this form of speculation by adding
microarchitectural events\footnote{A ``microarchitectural event'' can be thought of as an instruction that
produces (possibly operand-dependent) microarchitectural changes but no architectural changes.} for mispredicted
control-flow instructions and eventual rollback of a mis-speculated sequence of instructions; this would be
similar to various semantics developed in prior work~\cite{kopfcontracts, Spectector, InSpectre, blade}.

The implications of the above speculative semantics on our analysis are discussed in
Section~\ref{sec:security_goal}.

\subsection{Single Static-Assignment Form}
\label{sec:SSA}

Our analysis works with code written in single static-assignment (SSA) form. This is a popular and useful
abstraction that makes program data dependencies clear and variable identities unambiguous. Code is in SSA form
when any usage of a variable is reached by exactly one definition of that variable~\cite{ssa}.

In code with control flow, a variable's definition may depend on earlier control-flow decisions. In order to
make such code SSA-compliant, control flow-dependent definitions at join points are assigned to by
\phifunctions~\cite{ssa} which return an input based on the control-flow decision. We define the semantics of
the \phifunction as follows: suppose in some block $\block$ we have an instruction $\inst[\phi]$ that is a
\phifunction, $y = \phi(x_1,\dots,x_N)$. By definition, there are $N$ edges into $\block$, denoted $\edge[1],
\dots, \edge[N]$. The \phifunction is defined such that at any point in any execution, if $\edge[i]$ is the last
edge to reach $\inst[\phi]$, then $y = x_i$. An important consequence of the semantics of a \phifunction is that
every $x_i$ must be defined \emph{prior} to $\block$. This is simply because it must be a usable name by the
time the \phifunction is encountered.

As an example, the following non-SSA code on the left is transformed to produce SSA code on the right.

\begin{minipage}[H]{.4\linewidth}
    \begin{lstlisting}
    x = 0;
    if (...) {
      x = x + 1
    }
    print(x)
    \end{lstlisting}
\end{minipage}%
\hfill
\begin{minipage}[H]{.5\linewidth}
    \begin{lstlisting}
    $\codevar{x}[1]$ = 0;
    if (...) {
      $\codevar{x}[2]$ = $\codevar{x}[1]$ + 1
    }
    $\codevar{x}[3]$ = $\text{\straightphi}$($\codevar{x}[1]$,$\codevar{x}[2]$)
    print($\codevar{x}[3]$)
    \end{lstlisting}
\end{minipage}
\vspace{5pt}

\section{Setting and Security Goal}
\label{sec:security_goal}

The goal of \toolname is to efficiently prevent speculative leakage of sensitive data while maintaining a useful
security guarantee.

First and foremost, we define the points of potential information ``leakage.'' A \define{transmitter} is any
instruction whose execution exhibits operand-dependent hardware resource usage. Classic examples of transmitters
are loads and branches~\cite{kocher1996timing}. Depending on the microarchitecture, there may be
others~\cite{FPU_leaky,pandora,practical_doprogramming,oisa,USLH}. We say that the operands (data) passed to a
transmitter are \define{leaked}. Note that a transmitter may leak its operands fully or partially.

A \define{non-speculative transmitter} is one that appears in the program's non-speculative execution. A
\define{speculative transmitter} appears in the program's speculative execution. That is, it appears as an
operand-dependent microarchitectural event in the program's speculative semantics
(Section~\ref{sec:exec-traces}), which may or may not correspond to an instruction that architecturally retires.

We assume the standard attacker used in the constant-time programming
setting~\cite{stt,sdo,spt,blade,kopfcontracts}. Here, the attacker knows the victim program. The attacker
further sees a projection, or view, of the victim's non-speculative and speculative executions: namely,
\circled{1} the sequence of values taken by the program counter (PC), and \circled{2} the sequence of values
passed to transmitters.

With the above in mind, there are two main protection guarantees a speculative execution defense can
have~\cite{stt}. The first prevents speculatively-accessed data from being passed to speculative transmitters.
Enforcing this policy satisfies ``weak speculative non-interference''~\cite{kopfcontracts}, and is sufficient to
eliminate universal read gadgets and defend programs in sandbox settings~\cite{spectre_google}. As a result,
there has been significant interest in both hardware~\cite{stt,sdo,spectre_guard} and software~\cite{slh,blade}
defenses that provide said guarantee.

Unfortunately, such defenses are not comprehensive as there are still important applications---namely
constant-time cryptography---that \emph{non-speculatively} read and compute on sensitive data but can still leak
said data \emph{speculatively}~\cite{spt,kopfcontracts,stt,NDA,dolma}. To protect these programs, one requires a
defense with a broader protection guarantee: i.e., one that prevents both speculatively \emph{and
non-speculatively} accessed data from being passed to speculative transmitters. Defense mechanisms that meet
this guarantee provide \emph{complete} protection from speculative execution attacks but typically come at high
performance overhead, i.e., they are tantamount to delaying every transmitter's execution until they become
non-speculative or are squashed~\cite{NDA,SSLH,USLH}.

The aforementioned protections are (almost always) overly conservative because they implicitly treat \emph{all}
data as ``secret'' and deserving of protection. Yet, not all data is semantically secret. Software annotations
could directly convey what is and is not secret, but there are major disadvantages to programmer intervention.
For example, programmer intervention/expert labeling cannot be applied to legacy code already deployed. So, to
determine what data is secret without requiring expert labeling/intervention, we adopt a definition of
``secret'' proposed by SPT~\cite{spt},

\begin{definition}\label{defn:secret}
Data $x$ is \underline{secret} if there is no data flow from it to an operand of a non-speculative transmitter,
where data flow refers to flow through LLVM SSA variables and LLVM data memory.
\end{definition}

This definition is motivated by the constant-time programming model in which sensitive data is never passed to
non-speculative transmitters. The contrapositive of this is that if any data is passed to a non-speculative
transmitter, it is not sensitive, i.e. not ``secret''.

Definition~\ref{defn:secret} can be interpreted as enabling efficient implementations that satisfy
\define{generalized constant-time} (GCT)~\cite{kopfcontracts}. Denote the non-speculative and speculative
program semantics as $S^{\mathsf{nspec}}$ and $S^{\mathsf{spec}}$, respectively. For brevity, we also assume
these semantics encode the attacker's view, e.g., the set of transmitters. Given, a program $P$ and a policy
$\policy$ which defines what program variables are ``high'' (secret), $P$ satisfies GCT w.r.t.
$S^{\mathsf{nspec}}$, $S^{\mathsf{spec}}$ and $\policy$ if the following requirements hold.
\begin{enumerate}
\item Executions of $P$ on $S^{\mathsf{nspec}}$ satisfy non-interference w.r.t. $\policy$.  That is, attacker
observations of $P$'s execution, given $S^{\mathsf{nspec}}$, are independent of the values in
$\policy$.\label{vni}
\item Executions of $P$ on $S^{\mathsf{nspec}}$ that satisfy non-interference w.r.t. $\policy$ must also satisfy
non-interference on $S^{\mathsf{spec}}$ w.r.t. $\policy$.\label{sni}
\end{enumerate}
Requirement~\ref{sni} is referred to as \emph{speculative non-interference} or SNI for short~\cite{kopfcontracts,Spectector}.

Given this context, we can view \toolname as a function $P' = D(P; S^{\mathsf{nspec}}, S^{\mathsf{spec}})$ that
takes a program $P$ as input, produces a program $P'$ as output, and is parameterized by $S^{\mathsf{nspec}}$
and $S^{\mathsf{spec}}$. Suppose $P$ satisfies Requirement~\ref{vni}; it need not satisfy Requirement~\ref{sni}.
For a specified $S^\mathsf{nspec}$ and $S^{\mathsf{spec}}$, $D$ outputs a $P'$ that is functionally equivalent
to $P$ and now (additionally) satisfies Requirement~\ref{sni}, i.e., now satisfies SNI and therefore GCT.

Importantly, $D$ did not require $\policy$ as an input, but rather infers a policy $\policy[\decl]$ which is
sound w.r.t. $\policy$. That is, $\policy\subseteq\policy[\decl]$. This is possible because $D$ has access to
$P$, which already enforces $\policy$. At the same time, $\policy[\decl]$ will provide a basis for implementing
efficient protection. That is, if $\policy[\text{\scriptsize{All}}]$ denotes the policy (described above) that
treats all data as secret, we have that $\policy[\decl]\subseteq \policy[\text{\scriptsize{All}}]$ in theory and
$\policy[\decl]\subset \policy[\text{\scriptsize{All}}]$ in practice. This will allow us to more efficiently
protect programs without additional programmer intervention or labeling, beyond the program being written to
enforce non-speculative/vanilla constant-time execution.

\subsection{Semantics and Transmitters}
\label{sec:assumptions:goals}

$D$ is parameterized by $S^\mathsf{nspec}$ and $S^{\mathsf{spec}}$, which encode the execution semantics and
transmitters.

\paragraph{Semantics}
For security, our analysis assumes the semantics set forth in Section~\ref{sec:exec-traces}, in particular that
the speculative semantics is restricted to control-flow speculation that remains on the control-flow graph. This
is sufficient to protect non-speculatively accessed data in the presence of direct branches (similar to those
found in Spectre Variant 1). To block leakage due to other forms of speculation (e.g., indirect branches whose
\emph{targets} are predicted, as in Spectre variant 2), our analysis can adopt complementary defenses such as
``retpoline''.\footnote{See \url{https://support.google.com/faqs/answer/7625886}}

\paragraph{Transmitters}
Our analysis is flexible with respect to which instructions are considered transmitters. For the rest of the
paper, we assume loads are transmitters that can execute speculatively. We assume that branches and stores are
also transmitters, but only if they appear in the non-speculative execution. That is, we assume that branches
and stores do not change microarchitectural state in an operand-dependent way until they become non-speculative.
We note, this still allows for branch prediction; it just stipulates that said predicted branches only resolve
(and redirect execution) when they become non-speculative. To reiterate: these choices were not fundamental, and
the analysis can be modified to account for other transmitters and their speculative vs. non-speculative
behavior.

\section{Achieving Efficient Protection}
\label{sec:opportunities}

We now describe an analysis, dubbed \toolname, that enables low-overhead protection for ``secrets'' as given by
Definition~\ref{defn:secret}.

To understand our scheme's security and performance, we start by considering a secure but high-overhead
software-based protection. We will use the abstraction proposed by Blade~\cite{blade}, which introduces a
primitive called \protecc[\codevar{v}]. \protecc wraps a variable $\codevar{v}$ and delays its usage until it
becomes non-speculative (or \emph{stable}~\cite{blade}). Blade points out that, while current hardware does not
support \protecc, \protecc can be emulated today by introducing control-flow-dependent data
dependencies~\cite{slh}, speculation barriers, or a combination of the two. Regardless of how it is implemented,
executing \protecc incurs overhead by delaying an instruction's (and its dependents') execution. This is
especially pronounced when said instruction is on the critical path for instruction retirement (as is typical
with loads).

As discussed in Section~\ref{sec:security_goal}, we wish to protect non-speculatively accessed data. Thus, a
secure baseline defense must protect the operands of all transmitters that can execute speculatively. We express
this by wrapping said transmitters' operands with \protecc statements, placed immediately before each
transmitter and in the same block.\footnote{We note that this protection scope is broader than Blade's (which
only protects speculatively-accessed data), hence we don't compare to Blade further.} This approach is
tantamount to that of several recent defense proposals~\cite{SSLH,USLH,NDA}.

An example of our baseline is shown in Figure~\ref{fig:motiv-example-orig}, which depicts a program that
contains a transmitter inside a loop as well as at the exit point. The transmitters are denoted
\transmit[$\cdot$]. While this approach is secure, it is also expensive; the \pinktt{\protecc[x]} statements can
be encountered an unbounded number of times (depending on the semantics of the loop). Given this strict policy,
which effectively prevents transmitters from executing speculatively, nothing more can be done to improve
performance in this example.

However, if we instead consider Definition~\ref{defn:secret} and its implications, we can see that this program
contains unnecessary protection. From the discussion surrounding that definition, we saw that data which does
not meet the definition for ``secret'' does \emph{not} need to be protected from speculative leakage. This is
the manner in which we can reduce protection overhead. To aid in this process, we define a more useful concept
that is core to our work,

\begin{definition}\label{defn:known}A variable $x$ is considered \underline{known} when it is guaranteed to be
passed to a non-speculative transmitter (i.e., its value will inherently be revealed) or when its value can be
inferred from other known variables.
\end{definition}

We say that a variable can be ``inferred'' from other known variables if its value can be computed via a
polynomial-time algorithm\footnote{This is to admit computational assumptions. For example, knowing a plaintext
and its corresponding AES ciphertext should not give one the ability to ``infer'' the key!} from the values of
said other variables. We consider the set of known variables over time to be the attacker's
\define{non-speculative knowledge} (or just ``knowledge'' for short). One key addition made by
Definition~\ref{defn:known} is that a variable can be considered known not only when it is observed to be
non-secret, but even when it is guaranteed to \emph{eventually} become non-secret. That is, for the purposes of
our analysis, inevitable non-secrecy is as good as knowledge.

Our goal is to use Definition~\ref{defn:known} to derive a minimal set of locations at which to place \protecc
statements. For this, we need to define one more concept: the \emph{non-speculative knowledge frontier}
(\emph{knowledge frontier} for short) for each program variable. Intuitively, the knowledge frontier for a
variable $x$ represents the earliest points in the program such that, if the program's non-speculative execution
``crosses'' the knowledge frontier, $x$ will be known. We define it more precisely as,

\begin{definition}\label{defn:frontier} For any variable $x$, let $K_\allblocks(x)$ denote the set of blocks in
which $x$ is known. The \underline{knowledge frontier} of $x$, denoted $\frontier{x}$, is the smallest subset of
$K_\allblocks(x)$ such that for any $\block$ in $K_\allblocks(x)$, there is no path from $\entrynode$ to
$\block$ that does not contain some $\block'$ from $\frontier{x}$.
\end{definition}

Note that by this definition, the knowledge frontier for a variable in a given program is \emph{unique}. We can
reframe what is required of a protection mechanism to enforce SNI (with respect to the policy implied by
Definition~\ref{defn:secret}) in terms of the knowledge frontier,

\begin{property}(Frontier Protection Property for $x$) A placement of \protecc[x] statements enforces SNI with
respect to $x$ if and only if no speculative execution can transmit a function of $x$ before the non-speculative
execution crosses the knowledge frontier for $x$.
\label{prop:frontier}
\end{property}

One straightforward strategy to satisfy Property~\ref{prop:frontier} is to add \protecc statements only ``along
the knowledge frontier'' for each given variable, as opposed to at the site of each transmitter.

\paragraph{Example} Look again at the program in Figure~\ref{fig:motiv-example}. As mentioned, a naive
protection scheme places \pinktt{\protecc[x]} statements in the same blocks as all transmitters. However, if the
execution enters $\block[1]$ non-speculatively, it will necessarily (non-speculatively) enter either $\block[2]$
or $\block[3]$ next. Thus, the \blue{knowledge frontier} for \codevar{x} is \blue{$\block[1]$}. With this, we
can re-instrument the code with a \pinktt{\protecc[x]} inserted in $\block[1]$ only. Crucially, we have removed
the \pinktt{\protecc[x]} from the loop, which means that protection overhead will likely be amortized.

This is more aggressive than what is possible with the hardware-based defense SPT, on which
Definition~\ref{defn:secret} is based. Since SPT doesn't know the program's structure, it doesn't know if there
is a path from $\block[1]$ where \codevar{x} is not leaked non-speculatively, and hence it falls back to a
baseline protection that is tantamount to Figure~\ref{fig:motiv-example-orig}. As mentioned, an interesting
aspect of Definition~\ref{defn:known} is that it allows for variables to be known ``ahead of time'' (i.e. before
they are actually passed to a transmitter). This is a key difference between our approach and SPT's; knowledge
that is inevitable in the future is knowledge that can be exploited ``now.'' SPT on the other hand must wait to
observe transmitters retire before it treats its operands as known.

\begin{figure}[t]
    \centering
    \subfloat[][Code protected via prior work\label{fig:motiv-example-orig}]{%
    \includegraphics[height=4.2cm]{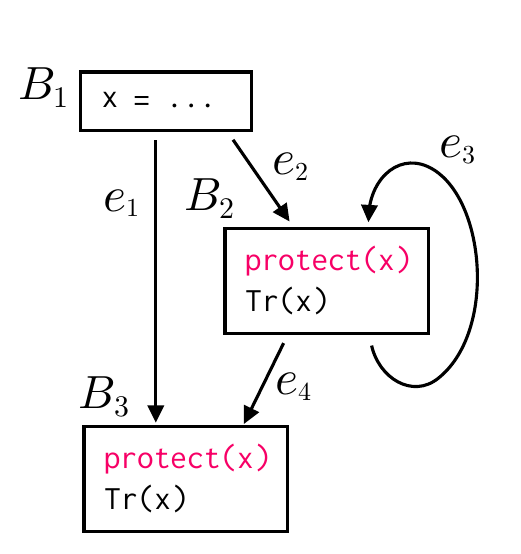}}
    \subfloat[][Code protected via our work\label{fig:motiv-example-decl}]{%
    \includegraphics[height=4.2cm]{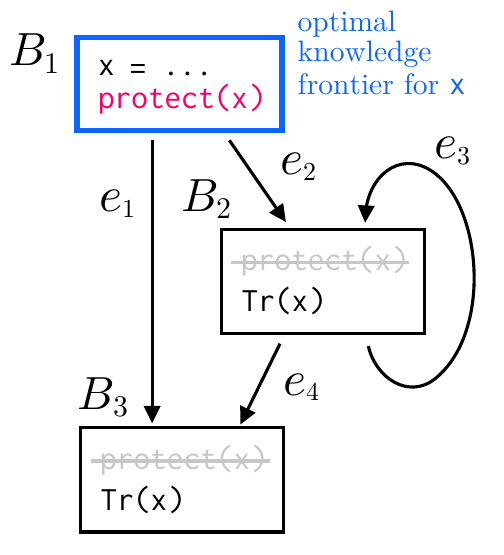}}
    \caption{Making the protected code (left) more efficient (right).}
    \label{fig:motiv-example}
\end{figure}

\section{Modeling Attacker Knowledge}
\label{sec:modeling-knowledge}

We will now define non-speculative attacker knowledge, which will be used to compute the non-speculative
knowledge frontier.

\subsection{Non-Speculative Knowledge}
\label{sec:non-spec-knowledge}

We model the attacker's knowledge at the granularity of variables, and we treat knowledge in a binary fashion;
the value of a variable is either ``fully'' known to an attacker or no function of the variable is known. Thus,
an attacker's knowledge is a subset of $\allvars$, and the full space of the knowledge of an attacker is
$\varlatt$.

Per Definition~\ref{defn:known}, in any trace, a variable is considered known at the point when (and any time
after) it is passed to a non-speculative transmitter (i.e., when its value is revealed), or at the point (and
any time after) its value can be inferred from other known variables. Rather than keeping track of knowledge
``temporally'' by tracking it over time (per trace), we can instead capture knowledge ``spatially'' by mapping
it onto the control-flow graph. We introduce the map $\alledgeknowns: \alledges \to \varlatt$ which represents
the distribution of knowledge over edges. We precisely define $\alledgeknowns$ as follows,

\begin{definition}\label{defn:KE} Take any edge $\edge \in \alledges$. We have $x \in \edgeknowns$ if and only
if for all traces $\trace \in \alltraces$ such that $\edge \in \trace$, $x$ is already known or is guaranteed to
become known every time the execution corresponding to $\trace$ traverses $\edge$. If $x \in \edgeknowns$, we
say ``$x$ is known on edge $\edge$''.
\end{definition}

There are two important clarifications we wish to make with respect to the above definition. First and foremost,
$x \in \edgeknowns$ does \emph{not} necessitate that $e$ is dominated by the definition of $x$ since the value
of $x$ may be inferable from other known variables, as discussed above. Second, the manner in which we define
knowledge and the way we intend to use it leave open the possibility for \define{vacuous knowledge}. A variable
is \define{vacuously known} if it is deduced to be known on a path on which it will never be defined. Crucially,
since such a variable is not defined on this path, it cannot be used. Thus, for the purposes of optimizing a
defense mechanism, this knowledge is inactionable. We will see in the next section that vacuous knowledge is an
important concept, particularly when analyzing \phifunctions.

\subsection{Instructions as Equations}
\label{sec:instructions}

We now define a set of relations that describe knowledge $\alledgeknowns$ with respect to a concrete program.

\subsubsection{\Nonphi instructions}

We first discuss how knowledge is propagated through \nonphi instructions. An instruction $\inst$ is said to be
\define{deterministic} if it can be represented as an equation of the form $y=f(x_1, \dots, x_N)$.\footnote{This
can be generalized to instructions with multiple outputs by writing down a separate equation for each output.}
Control-flow instructions, by convention, don't have an output. Loads and stores are not considered
deterministic since our current analysis does not model the contents of memory. We define $\instoutput = y$ and
$\instinputs = \{x_1, \dots, x_N\}$.

Deterministic instructions are of interest since knowing all but one of their operands/results enables deduction
of the remaining one. We say an instruction $I$ is \emph{forward solvable} if, for all concrete assignments to
$x_i \in \instinputs$, we have a unique solution for $y$. We say $I$ is \emph{backward solvable} if for any $x_i
\in \instinputs$, for all concrete assignments to $\instoutput$ and all $x_{i \neq j} \in \instinputs$, we have
a unique solution for $x_i$. All deterministic instructions (e.g., \texttt{add}, \texttt{sub}, \texttt{mul}) are
forward solvable; not all are backwards solvable.\footnote{For example, $y = x_1 \times x_2$ is deterministic
and forward solvable, but not backwards solvable; if one operand is 0, the output is 0 regardless of the other
operand's value.} Exploiting the solvability of instructions is how we achieve propagation of knowledge through
computations as motivated in Section~\ref{sec:opportunities}.

An important consequence of working with programs expressed in SSA form is that the equations that define
variables are \emph{unique}. Crucially (and perhaps counter-intuitively) this implies the equations associated
with instructions are exploitable \emph{anywhere} in the control-flow graph, even if the associated instruction
is not encountered in a given trace or is unreachable in general. When analyzing a non SSA-form program, there
may be multiple definitions for any given variable, and thus you need to consider only the definitions that
apply to the locale you are in. Attempting to analyze a non-SSA program while keeping track of which definitions
apply where implicitly converts the program to SSA form. We will see the benefits of the global view of
equations in the examples from Section~\ref{sec:two-examples}.

\subsubsection{\phifunctions}

We now discuss \phifunctions. Before we begin, recall that the PC is public to the attacker due to assumptions
made by the constant-time programming model (Section~\ref{sec:security_goal}). \emph{Since the PC is public, it
can be considered known at all points in time in the non-speculative execution}.

Because the PC is known, we can treat \phifunctions as being forward solvable. If at any point all of a
\phifunction's inputs are known, then regardless of which one gets assigned to the output, the output must be
known as well. Consider a \phifunction $\inst[\phi]$ of the form $y = \phi(x_1,\dots,x_N)$ in a block $\block$
with input edges $\edge[i]$ and output edges $\edge[o]$. The semantics of $\inst[\phi]$ are such that $y = x_i$
if $\edge[i]$ is taken to reach $\block$. Now, suppose that $x_1,\dots,x_N$, are known on some edge $\edge$. If
$\edge$ is in a trace containing $\inst[\phi]$, we know $y$ because we know which $x_i$ is assigned to $y$
(because the PC is known) and we know each $x_i$. Note that if $\inst[\phi]$ is \emph{not} in a trace containing
$\edge$, then knowledge of $y$ is vacuous; it cannot be used in a meaningful way.

There is an additional, more subtle version of forward solvability when considering \phifunctions. Consider
again $\inst[\phi]$ as defined before. If each input $x_i$ to the \phifunction is known on its respective edge
$\edge[i]$, $y$ is known on all $\edge[o]$. This is again due to the assumption that the PC is known; the input
edge used to arrive at $\block$ is known and thus we will know which $x_i$ is assigned to $y$.

Note that \phifunctions are \emph{not} backward solvable; knowing the output and all but one of the inputs
doesn't necessarily reveal the last input. That said, there \emph{is} a causal relationship we can exploit in
the backwards direction. Using the definition/semantics of $\inst[\phi]$ from before, suppose that $y$ is known
on all output edges $\edge[o]$. Then every $x_i$ is known on its respective edge $\edge[i]$. The justification
is as follows: suppose $\edge[i]$ is traversed. Then $y = x_i$ and we know for which $x_i$ this holds. Now, we
must leave $\block$ through some $\edge[o]$, and $y$ is known on every $\edge[o]$. Thus, in this scenario, $x_i$
is known.

\subsubsection{Knowledge propagation theorems}

We summarize the previous discussion with a series of theorems that describe relationships on knowledge. The
proofs of these theorems can be found in Appendix~\ref{app:proofs}. We start with theorems that describe the
knowledge available to every edge in isolation,

\begin{theorem}\label{thm:KE_transmitter}Consider an instruction $I$ of the form $\transmit[x]$ in some block
$B$ with output edges $\edge[o]$. For all $\edge[o]$, $x \in \edgeknowns[o]$. \end{theorem}

\begin{theorem}\label{thm:KE_fwd_intra}Take any edge $\edge \in \alledges$. Consider a forward solvable
instruction $\inst$ from anywhere in the control-flow graph, and suppose it is of the form $y=f(x_1, \dots,
x_N)$. If $x_i \in \edgeknowns$ for all $x_i$, then $y \in \edgeknowns$.
\end{theorem}

\begin{theorem}\label{thm:KE_bwd_intra}Take any edge $\edge \in \alledges$. Consider a backward solvable
instruction $\inst$ from anywhere in the control-flow graph, and suppose it is of the form $y=f(x_1, \dots,
x_N)$. For any $j \in \set{1,\dots,N}$, suppose we have all $x_{i\neq j} \in \edgeknowns$ as well as $y \in
\edgeknowns$. Then $x_j \in \edgeknowns$.
\end{theorem}

The following theorems describe the relationship of knowledge between edges in the general case,

\begin{theorem}\label{thm:KE_fwd_inter} Consider a block $\block$. If for some variable $v$ we have $v \in
\bigcap \edgeknowns[i]$ for all realizable $\edge[i]$ in $\inedgeset$, then we have $v \in \edgeknowns[o]$ for
every $\edge[o] \in \outedgeset$.
\end{theorem}

\begin{theorem}\label{thm:KE_bwd_inter} Consider a block $\block$. If for some variable $v$ we have $v \in
\bigcap \edgeknowns[o]$ for all realizable $\edge[o]$ in $\outedgeset$, and if $v$ is not defined in $\block$,
then we have $v \in \edgeknowns[i]$ for every $\edge[i] \in \inedgeset$.
\end{theorem}

An important thing to notice about Theorem~\ref{thm:KE_bwd_inter} is that we mandate $v$ is not defined in
$\block$. This is because if $v$ is \emph{not} defined in $\block$, then nothing in $\block$ can change $v$'s
status in terms of knowledge; this is not true if $v$ is defined in $\block$. That said, the theorem does
\emph{not} prevent variables from being known prior to their definition; they just need to be inferable (via the
other theorems).

The following theorems describe the relationship of knowledge between edges in the special case that they are
connected via a \phifunction.

\begin{theorem}\label{thm:KE_phi_fwd} Let $\inst[\phi]$ denote a \phifunction of the form $y =\allowbreak
\phi(x_1,\allowbreak\dots,\allowbreak x_N)$ in block $\block$ with $N$ input edges $\edge[1],\dots,\edge[N]$.
The semantics of $\inst[\phi]$ are such that $y = x_i$ if $\edge[i] \in \inedgeset$ is traversed to reach
$\block$. If for all realizable $\edge[i]$ we have $x_i \in \edgeknowns[i]$, then for every output edge
$\edge'$, $y \in \edgeknowns[\edge'][]$.
\end{theorem}

\begin{theorem}\label{thm:KE_phi_bwd} Let $\inst[\phi]$ denote a \phifunction of the form $y =\allowbreak
\phi(x_1,\allowbreak\dots,\allowbreak x_N)$ in a block $\block$ with input edges $\edge[1],\dots,\edge[N]$. The
semantics of $\inst[\phi]$ are the same as in Theorem~\ref{thm:KE_phi_fwd}. If for all realizable output edges
$\edge'$ we have $y \in \edgeknowns[\edge'][]$, then for every $\edge[i]$ we have $x_i \in \edgeknowns[i]$.
\end{theorem}

Note that for Theorems~\ref{thm:KE_transmitter}--\ref{thm:KE_bwd_intra}, the realizability of an edge is not
important to our model of knowledge. This is because the claims are about the \emph{local} knowledge associated
with each edge individually, and unrealizable edges only contain vacuous knowledge. On the other hand,
Theorems~\ref{thm:KE_fwd_inter}--\ref{thm:KE_phi_bwd} make use of the realizability of edges since we are
considering the relationships \emph{between} edges.

\subsection{Examples Computing Knowledge/Frontiers}
\label{sec:two-examples}

\begin{figure}[t]
    \centering
    \subfloat[][Example 1\label{fig:program-B}]{%
    \includegraphics[height=5.6cm]{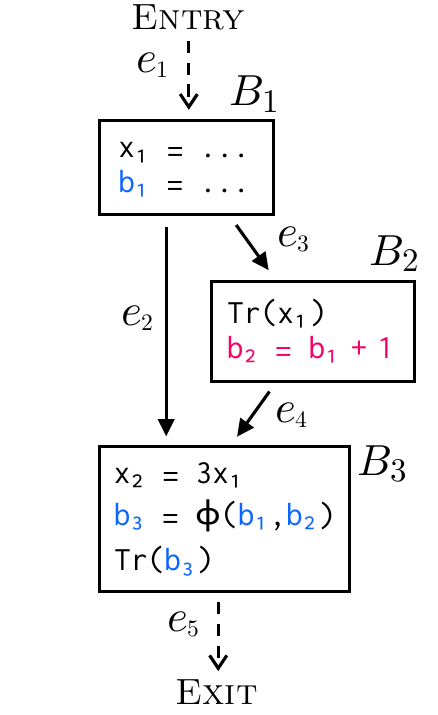}}
    \hspace{20pt}
    \subfloat[][Example 2\label{fig:program-A}]{%
    \includegraphics[height=5.6cm]{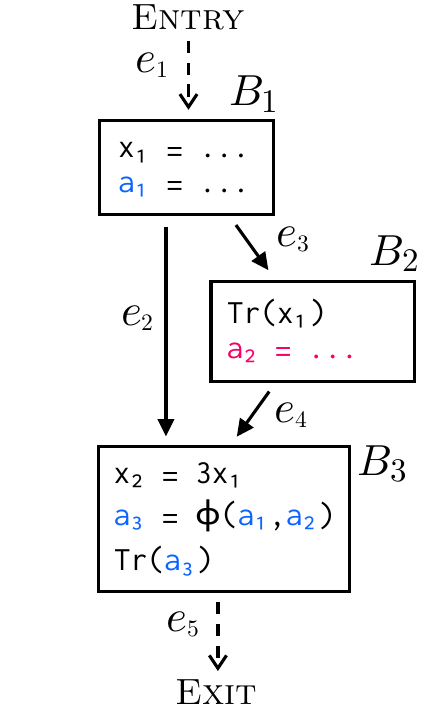}}
    \caption{Examples used to illustrate knowledge over edges ($\alledgeknowns$).}
    \label{fig:K-example}
\end{figure}

Figure~\ref{fig:K-example} shows an example of how to compute knowledge (and then the knowledge frontier) with
the relations from Section~\ref{sec:instructions}.

First consider Figure~\ref{fig:program-B}. There are two paths through the control-flow graph; we'll assume they
both correspond to realizable traces. Starting off, we have \transmit[\codevar{b}[3]] in $\block[3]$ which leads
to knowing $b_3$ on edge $\edge[5]$ (Theorem~\ref{thm:KE_transmitter}) and subsequently $b_1$ on $\edge[2]$ and
$b_2$ on $\edge[4]$ (Theorem~\ref{thm:KE_phi_bwd}). The equation $b_2 = b_1 + 1$ and
Theorem~\ref{thm:KE_fwd_intra} enable us to deduce $b_2$ everywhere $b_1$ is known (and vice versa by
Theorem~\ref{thm:KE_bwd_intra})---similarly for $x_1$, $x_2$ and $x_2 = 3x_1$. Since both $b_1$ and $b_2$ are
known on $\edge[2]$, $\edge[3]$, and $\edge[4]$, by Theorem~\ref{thm:KE_fwd_intra}, $b_3$ is also known on
$\edge[2]$, $\edge[3]$, and $\edge[4]$.

By continuing to apply the theorems, we get $\edgeknowns[1] = \emptyset$, $\edgeknowns[2] = \set{b_1, b_2,
b_3}$, $\edgeknowns[3] = \edgeknowns[4] = \set{b_1, b_2, b_3, x_1, x_2}$, and $\edgeknowns[5] = \set{b_1, b_2,
b_3}$.

This provides a basis for the frontier of $x_1$ and $x_2$ to be $\set{\block[2]}$. More importantly, it means
the frontier for $b_1,b_2,b_3$ is $\{\block[1]\}$. A program that non-speculatively enters $\block[1]$ will
inherently transmit all three. (We detail more precisely how to compute the frontier given $\alledgeknowns$ in
Section~\ref{sec:knowledge-frontier}.) This will enable more efficient protection: to enforce
Property~\ref{prop:frontier}, it is sufficient to add \protecc statements for $b_1,b_2,b_3$ solely in
$\block[1]$.

Next consider Figure~\ref{fig:program-A}, which is almost the same as Figure~\ref{fig:program-B} except for two
changes: first, every $b_i$ is replaced with $a_i$ for clarity; second (and most importantly), the equations for
$a_2$ and $b_2$ differ. We use \texttt{\codevar{a}[2]$\;$=$\;\dots$} to represent a non-deterministic
instruction. Since we no longer have an equation relating $a_2$ to $a_1$ (as we did with $b_2 = b_1 + 1$ from
before), the knowledge settles to $\edgeknowns[1] = \emptyset$, $\edgeknowns[2] = \set{a_1}$, $\edgeknowns[3] =
\edgeknowns[4] = \set{x_1, x_2, a_2}$, and $\edgeknowns[5] = \set{a_3}$. From this we can deduce that the
frontier for $a_1$ is $\emptyset$; the frontier for $x_1$, $x_2$ and $a_2$ is $\set{\block[2]}$; the frontier
for $a_3$ is $\set{\block[3]}$. This matches our security goal: it is unsafe to hoist \protecc statements above
any variable's definition because knowing $a_1$ is not the same as knowing $a_2$ and hoisting a \protecc would
enable an attacker to selectively learn both through \transmit[\codevar{a}[3]].

\subsection{Approximating Non-Speculative Knowledge}
\label{sec:approx-knowledge}

\begin{figure}[t]
\centering
\begin{minipage}[H]{.28\linewidth}
    \begin{lstlisting}[frame=r, linewidth=2.1cm]
if (q) {
  @Tr(x)@
}
if ($\text{!}$q) {
  $\;\!$@Tr(x)@
}
    \end{lstlisting}
\end{minipage}%
\hspace{13pt}
\begin{minipage}[H]{.65\linewidth}
    \includegraphics[width=\linewidth]{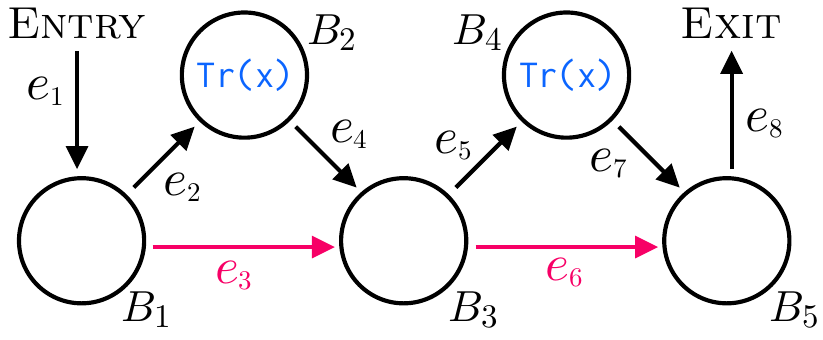}
\end{minipage}
\caption{A program and its control-flow graph. The branches are anti-correlated, thus the trace
$\{\edge[1],\pinkm{\edge[3]},\pinkm{\edge[6]},\edge[8]\}$ is not realizable.
}
\label{fig:approx}
\end{figure}

Precisely computing $\alledgeknowns$ using the theorems from the previous section is generally intractable since
it relies on knowing whether any given edge is realizable. We can instead attempt to compute an
\emph{approximation} of $\alledgeknowns$, denoted as $\alldfvs$. We consider our approximation \define{sound} if
it does not over-estimate an attacker's true knowledge; we consider it \define{imprecise} if it under-estimates
it. To compute $\alldfvs$, we make the assumption that any path through the control-flow graph corresponds to a
realizable trace. More specifically, for any edge $\edge = (\block[i], \block[j])$, we assume that any edge
$\edge' \in \inedgeset[i]$ may have been traversed prior to it, and any edge $\edge'' \in \outedgeset[j]$ may be
traversed after it. Looking back at Definition~\ref{defn:KE}, by increasing the number of traces we consider
$\edge$ to have been part of, we are (potentially) reducing the size of $\dfv$; i.e. $\dfv \subseteq
\edgeknowns$. Thus, this approximation method (potentially) loses precision, but it maintains soundness.

\paragraph{Example} Look at Figure~\ref{fig:approx}. The only possible traces are $\trace[1] = \set{\edge[1],
\edge[2], \edge[4], \pinkm{\edge[6]}, \edge[8]}$ and $\trace[2] = \set{\edge[1], \pinkm{\edge[3]}, \edge[5],
\edge[7], \edge[8]}$. (We've colored edges $\pinkm{\edge[3]}$ and $\pinkm{\edge[6]}$ to correspond with the
figure.) The paths $\{\allowbreak\edge[1], \allowbreak\pinkm{\edge[3]}, \allowbreak\pinkm{\edge[6]},
\allowbreak\edge[8]\}$ and $\{\allowbreak\edge[1], \allowbreak\edge[2], \allowbreak\edge[4],
\allowbreak\edge[5], \allowbreak\edge[7], \allowbreak\edge[8]\}$ do not correspond to realizable traces since
they imply \codevar{q} is both true and false. Ideally then, we must have that $x$ is known on $\edge[3]$, i.e.
$\edgeknowns[\pinkm{\edge[3]}][] = \set{x}$, since at this point the execution \emph{must} take $\edge[5]$ and
encounter \bluett{\transmit[x]}. Making the simplifying assumption (above) to compute $\dfv$, however, we cannot
assume whether we will take $\edge[2]$ vs. $\pinkm{\edge[3]}$; nor can we assume that we will take $\edge[5]$
vs. $\pinkm{\edge[6]}$ (for whichever of $\edge[2]$ or $\pinkm{\edge[3]}$ we took). In other words, the analysis
cannot conclude whether we will encounter \bluett{\transmit[x]} and thus deduces that $x$ is \emph{not} known on
$\pinkm{\edge[3]}$; i.e. $\dfv[\pinkm{\edge[3]}][] = \emptyset$. See that $\dfv[\pinkm{\edge[3]}][] \subset
\edgeknowns[\pinkm{\edge[3]}][]$; we've lost precision in order to gain tractability, but we have not sacrificed
soundness.
\section{\toolname Approaches}
\label{sec:design}

In this section, we describe the approach used by \toolname to compute and utilize non-speculative attacker knowledge.

At a high level, our analysis first computes $\dfv$ for efficiency reasons (Section~\ref{sec:dfa}), but invokes
more sophisticated analyses to compute $\edgeknowns$ on specific edges when doing so is deemed profitable
(Section~\ref{sec:symex}). Later subsections then detail how to use the edge-based knowledge to compute the
knowledge frontier (Section~\ref{sec:knowledge-frontier}) and instrument protection
(Section~\ref{sec:relaxation}). Section~\ref{sec:implementation} goes over lower-level details of all of the
above.

\subsection{Computing Knowledge via a Data-Flow Analysis}
\label{sec:dfa}

Our ideas from Section~\ref{sec:modeling-knowledge} map naturally to a \emph{data-flow analysis}~\cite{dragon-book}.

Data-flow analyses work by assigning \emph{data-flow values} to every point in the control-flow graph. They
iteratively apply local rules known as the \emph{data-flow equations} to build up and combine these data-flow
values. When this process has converged (i.e. the size of the data-flow values stagnates), what remains is a
\emph{data-flow solution}.

A data-flow analysis makes the assumption that any path through the control-flow graph is a potential path the
execution may take; it approximates $\alltraces$ with $\allpaths$. Since the definition of $\alldfvs$ relies on
the same assumption, we can formulate a data-flow analysis such that the solution is \emph{exactly} $\alldfvs$.
The specific data-flow equations we use are discussed in Section~\ref{sec:dfa-pass}. By construction, the result
of our ideas applied to Figure~\ref{fig:approx} (as discussed in Section~\ref{sec:approx-knowledge}) is
precisely what our data-flow analysis would yield.

\subsection{Improving Precision via Symbolic Execution}
\label{sec:symex}

There are of course drawbacks to approximating $\alledgeknowns$ as done by the data-flow analysis (see
Section~\ref{sec:approx-knowledge}). Namely, under-estimating an attacker's knowledge causes us to over-estimate
the amount of protection needed. To avoid this, we need a way to analyze programs in a way that can deduce
whether a trace is unrealizable.

To that end, we also selectively employ symbolic execution~\cite{symex-survey}. In a nutshell, symbolic
execution involves executing a program with ``symbolic'' values; variables are mapped to symbolic expressions
rather than concrete values (at least, in the cases when the concrete value cannot be deduced unambiguously).
Expressions associated with all variables are derived by the instructions encountered during the symbolic
execution. Upon reaching a branch, a symbolic execution engine will traverse both paths separately. Every path
through will have associated with it ``path constraints'', which are a set of symbolic expressions implied by
the set of taken branches.

Our framework uses symbolic execution to answer the following question,

\begin{question}\label{q:symex} Given some region $\region$ of the control-flow graph and given some variable
$x$, does there exist a path through $\region$ (that corresponds to a realizable trace) upon which $x$ is
\emph{not} transmitted?
\end{question}

Recall that a data-flow analysis conservatively answers this question by assuming that if a path exists, it is
part of a realizable trace. By considering the semantics of the instructions and the branch conditions, symbolic
execution can try and answer the question less conservatively; though we stress that it does so in a sound
manner since it needs to \emph{prove} that the path cannot be taken.

We can look back at Figure~\ref{fig:approx} to see how symbolic execution can succeed where the data-flow
analysis fails. The candidate region we consider is the entirety of the program. Recall that the problematic
path was $\path' = \{\edge[1], \pinkm{\edge[3]}, \pinkm{\edge[6]},\edge[8]\}$, which again does \emph{not}
correspond to a realizable trace. When the tool considers the path $\path'$, the path constraints will contain
\emph{both} $q = \text{false}$ and $q \neq \text{false}$. These contradictory statements mean that no
non-speculative execution could ever traverse such a path. Thus, symbolic execution will conclude that
\blue{\transmit[x]} is unavoidable (i.e. that the answer to Question~\ref{q:symex} is ``no'' for this region and
variable $x$) meaning $x$ is guaranteed knowledge at (among other places) the program's entry point. Thus, we
have achieved a more precise result than the data-flow analysis.

The details of how and when we utilize symbolic execution to answer Question~\ref{q:symex} are given in
Section~\ref{sec:symex-pass}.

Note that the symbolic execution cannot be used on its own; the results of the data-flow analysis are a
prerequisite. For symbolic execution to work, we will need to instrument the code to indicate what variables are
known at various points, and the data-flow analysis is precisely what provides this information. While it is
certainly possible to formulate the entire analysis in terms of symbolic execution, this would certainly not
scale to larger programs as well as a data-flow analysis would.

\subsection{Computing the Knowledge Frontier}
\label{sec:knowledge-frontier}

As discussed in Section~\ref{sec:opportunities}, Property~\ref{prop:frontier} is the key to finding a minimal
yet sufficient set of \protecc statements needed to secure a program. To that end, we need to compute the
knowledge frontier given the results of the data-flow analysis and/or symbolic execution.

The first step to computing the knowledge frontier is to map knowledge from edges to basic blocks; we want the
map $\allblockknowns: \allblocks \to \varlatt$. For any block $\block$, we define $\blockknowns = \bigcap
\dfv[\edge'][]$ for all $\edge' \in \outedgeset$. Recall that $\alldfvs$ is the result of the data-flow
analysis. The results from the symbolic execution constitute additions to $\allblockknowns$.
Question~\ref{q:symex} is associated with some candidate variable $x$ and some candidate region $\region$ of the
control-flow graph. If, when given these, the symbolic execution tool answers ``no'' to Question~\ref{q:symex},
then $x \in \blockknowns$ for all $\block \in \region$.

With $\allblockknowns$ in hand, we can compute the frontiers for all variables. For any variable $x$, we first
over-estimate $\frontier{x}$ by adding all $\block \in \allblocks$ such that $x \in \blockknowns$. Then, for any
$\block \in \frontier{x}$, if $x$ is known in all of $\block$'s predecessor blocks, we remove $\block$ from
$\frontier{x}$. This is done using another data-flow analysis. After this, $\frontier{x}$ represents the precise
knowledge frontier for $x$.

Consider an arbitrary program function \texttt{f}. For any variable $x$, if its frontier $\frontier{x}$ is the
entry block of \texttt{f}, we say $x$ is \define{fully declassified}. If all variables transmitted by \texttt{f}
are fully declassified, then \texttt{f} itself is fully declassified.

\subsection{Adding Protection}
\label{sec:relaxation}

Once we have the knowledge frontier $\allfrontiers$, we can place protection using a simple strategy: for every
variable $x$ that is transmitted, we place \protecc[x] statements all along its frontier $\frontier{x}$. We
refer to this approach as ``enforcing the knowledge frontier''. This is a straightforward method to satisfy
Property~\ref{prop:frontier}. In terms of hoisting \protecc statements as high as possible, it is also optimal.

For simplicity, we perform our analysis at the granularity of functions. To that end, we now discuss two
strategies---\emph{callee enforcement} and \emph{caller enforcement}---that an analysis can use to instrument
protection when considering calls between functions. We use both of these in our final implementation.

\paragraph{Callee enforcement}
Callee enforcement is a straightforward but potentially high overhead strategy. Suppose we have a program
function \texttt{f}, which calls \texttt{g}. With callee enforcement, \texttt{f} and \texttt{g} are analyzed and
instrumented with protections in isolation. That is, \texttt{g} (the callee) is protected regardless of
knowledge in the caller context \texttt{f}. This approach allows us to ignore the interactions between functions
and have every function focus on enforcing its own frontier in isolation. While simple, this strategy may
overprotect the callee. For example, all variables in \texttt{g} that require protection may be known before
\texttt{g} is called.

\paragraph{Caller enforcement}
To reduce overhead stemming from callee protection, we now consider an alternative approach called caller
enforcement. The high-level idea is that, if certain conditions about the callee are met, the caller can
abstractly view the callee as any other transmitter. We call these \emph{pseudo transmitters}: program-level
functions that leak (a subset of) their arguments but no internal data. More precisely, consider a function
$f(x_1,\dots,x_N)$ that non-speculatively leaks some non-empty subset of its arguments $x_1', \dots, x_M'$ (with
$M \leq N$) (possibly) along with some other variables $v_1, \dots, v_K$.

The function $f$ is a \define{pseudo transmitter} if and only if: \circled{1} $x_1', \dots, x_M'$ and $v_1,
\dots, v_K$ are all fully declassified in $f$; \circled{2} knowledge of $x_1', \dots, x_M'$ is sufficient to
infer every $v_j$; \circled{3} all other functions called by $f$ are themselves pseudo transmitters. The key
insight is that the information leaked by a pseudo transmitter can be understood strictly in terms of its
arguments, which means we can reason about its protection in its calling contexts.\footnote{We do not have this
luxury with functions that are not pseudo transmitters even if they are fully declassified. Their internal
leakages mandate that every call to the function creates its own personal frontier.} Rather than protecting
every function call, we can enforce the frontier of the \emph{arguments} that are (non-speculatively) leaked by
each function call. If calls to the same function have arguments that are connected via data flow, we can
exploit this to protect their common frontier.

To see the benefits of caller enforcement, consider the following example. Suppose some program function
\codevar{f} makes two calls to another function \func{g}[][x] that is a pseudo transmitter which leaks its
argument \texttt{x} and no internal variables. Suppose the first call to \texttt{g} is \func{g}[1][x] and the
second call is \func{g}[2][$\texttt{x}'$] where \texttt{$\text{x}'=\text{x} + \text{1}$} (the subscripts are
used to distinguish the calls). Suppose further that \codevar{g}[1] dominates \codevar{g}[2]. Both calls leak
their arguments, so naively, the frontier for \texttt{x} is the calling context of \codevar{g}[1] while the
frontier of $\texttt{x}'$ is the calling context of \codevar{g}[2]. However, since \texttt{x} and $\texttt{x}'$
are equivalent in terms of knowledge (due to the backward solvability of \texttt{$\text{x}'=\text{x} +
\text{1}$}), and since \codevar{g}[1] dominates \codevar{g}[2], we can promote the frontier of $\texttt{x}'$ to
that of \texttt{x}. Enforcing the frontier of \texttt{x} is sufficient to protect the call to
\func{g}[2][$\texttt{x}'$]. Thus one \protecc can cover two function calls. This is strictly better than callee
enforcement which would have protected \codevar{g}[1] and \codevar{g}[2] separately. Note that being a pseudo
transmitter is both sufficient and necessary for a function to be caller enforced.

\section{Implementation Details}
\label{sec:implementation}

We now give details for how the ideas from Sections~\ref{sec:modeling-knowledge} and \ref{sec:design} are
implemented.

\toolname's main components follow the sketch given in Section~\ref{sec:design}, and we adopt the following
strategy for combining the data-flow analysis, symbolic execution and protection steps together. First, we apply
a data-flow analysis (Section~\ref{sec:dfa-pass}). Second, if the data-flow analysis results are suboptimal
(i.e. we have functions that are not fully declassified), we analyze those functions using the symbolic
execution tool \klee (Section~\ref{sec:symex-pass}). After running one or both analyses, we place protections
(Section~\ref{sec:relaxation-pass}). We apply all 3 passes in that order to every function.

Our framework contains both intra- and inter-procedural aspects. By design, \klee is both intra- and
inter-procedural.\footnote{That is, \klee will symbolically execute the top level function we provide it as well
as any callees, and so on recursively.} Thus, our symbolic execution pass adopts both these traits. The
data-flow analysis and protection pass are primarily intra-procedural but are performed on functions in a
specific order so as to pass non-speculative knowledge computed in a callee up the call graph. Specifically, we
apply the analysis in the order of callees then callers, i.e., work up the call graph starting at the
leaves.\footnote{This approach assumes that the call-graph is acyclic, i.e. there is no recursion. In cases
where this does not hold, we can extend our method to repeatedly analyze functions and stop when no new
information is derived. However, since we do not encounter recursion in our benchmarks, we omit implementation
and further discussion.}

\subsection{The Data-Flow Pass}
\label{sec:dfa-pass}

As mentioned in Section~\ref{sec:dfa}, we can formulate a data-flow analysis that computes $\alldfvs$. We use
control-flow edges as the program points of interest, and the data-flow value for a given edge $\edge$ is
precisely $\dfv$. The data-flow rules are based on the theorems from Section~\ref{sec:instructions}.

Running the data-flow analysis is done in two steps. First, we perform an LLVM control-flow-graph-level
transformation to ensure that loops are correctly modeled. Second, we initialize data-flow values along all
program edges and iteratively apply the data-flow rules until $\dfv$ is constructed.

\begin{figure}[t]
    \centering
    \subfloat[][Original loop\label{fig:pita-loop-2}]{%
    \includegraphics[height=4.3cm]{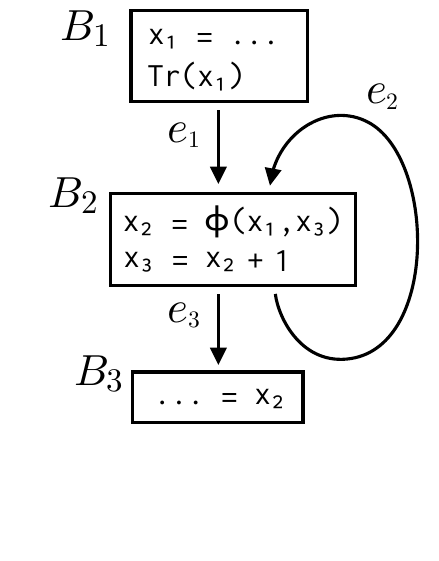}}
    \hspace{0.9cm}
    \subfloat[][Partially expanded loop\label{fig:fixed-loop-2}]{%
    \includegraphics[height=4.3cm]{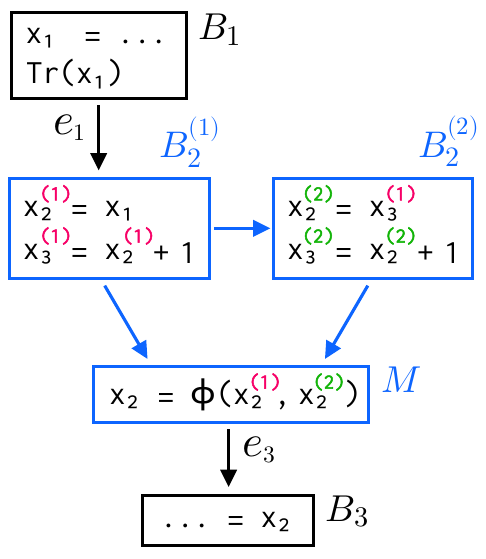}}
    \caption{On the left is a loop which creates an inductive relationship between \codevar{x}[1] and \codevar{x}[2]. Partially expanding it as shown on the right allows our data-flow analysis to capture this relationship.}
    \label{fig:loop-expansion}
\end{figure}

\paragraph{Loop transformation}
In the presence of loops, the data-flow rules can fail to capture inductive relationships (e.g., loop-carry
dependencies). For example, in Figure~\ref{fig:pita-loop-2}, they would be unable to deduce that knowledge of
\codevar{x}[1] in $\block[1]$ implies knowledge of \codevar{x}[2] in $\block[3]$. To remedy this, prior to the
data-flow analysis, we perform \define{partial loop expansion}. This procedure takes a control-flow graph with a
loop and transforms it to be acyclic. We accomplish this by duplicating the loop body and removing the back
edge. Unlike full loop unrolling, we only keep two cases: the initial case and the inductive case. We note that
although this procedure destroys the ``correctness'' of the program in terms of the values computed, it
preserves the relationships (between variables) that are relevant to modeling knowledge. This technique is
crucial for analyzing the benchmarks in Section~\ref{sec:eval}. We present more details about the partial loop
expansion procedure in Appendix~\ref{app:loop-expansion}.

\paragraph{Data-flow initialization and evaluation}
Once the program control-flow graph is transformed to account for loops, we proceed to run the data-flow
analysis.

We start by initializing data-flow values. We look at all the transmitters in the program. For all $\block \in
\allblocks$, and for every $\edge' \in \outedgeset$, we initialize $\dfv[\edge'][]$ to be the set of all
variables in $\block$ that are directly passed to a transmitter. This is a direct application of
Theorem~\ref{thm:KE_transmitter}. We also utilize some inter-procedural initialization, but with a unidrectional
flow of information.\footnote{In theory, there is benefit to augmenting the analysis to have callers send
\emph{and} receive information to/from callees. However, it complicates the analysis, and these opportunities
don't arise in our benchmarks. Thus, we leave this for future work.} Suppose we are analyzing a program function
\codevar{f} that makes a call to some other function \func{g}[][x] in block $\block$. If previous analysis of
\codevar{g} had concluded that the frontier of \codevar{x} within \codevar{g} was the entry block of
\codevar{g}, then in our analysis of \codevar{f} we may add $x$ to $\dfv$ for all $\edge \in \outedgeset$. This
step relies on \codevar{g} having been analyzed prior to \codevar{f}, which is why we must perform our analysis
in the order of callees then callers as mentioned before.

We then apply Theorems \ref{thm:KE_fwd_intra}-\ref{thm:KE_phi_bwd} iteratively until we have achieved
convergence; that is, until the size of $\alldfvs$ has reached a fixed point. Recall that the range of of our
data-flow values is $\varlatt$, which is a powerset lattice of finite height. Our data-flow rules never decrease
the size of the data-flow value, i.e. the values can never move down the lattice. Thus, there is a limit to how
much the data-flow values can grow before stagnating. This guarantees that our analysis will converge in finite
time.

\subsection{The Symbolic Execution Pass}
\label{sec:symex-pass}

We apply the symbolic execution pass on functions that are not fully declassified; our goal is to deduce
additional attacker knowledge as detailed in Section~\ref{sec:symex}. To that end, we invoke \klee~\cite{klee}
to answer Question~\ref{q:symex} with respect to a region $\region$ and variable $x$. We are interested in
analyzing regions that encompass all transmitters. Let $\transmitblocks$ be the set of all blocks which contain
at least one transmitter. Any region $\region$ such that $\transmitblocks \subseteq \region$ is a candidate
region for the symbolic execution pass. We can find these automatically by looking at all $\block \in
\allblocks$ and keeping the ones that collectively dominate every transmitter. Then, the regions defined by each
of these is a candidate region. Furthermore, any variable which is known in a block that contains a transmitter
(even if that variable is itself not transmitted in that block) is a candidate variable.\footnote{If the
data-flow analysis has deduced a candidate variable to be known throughout the candidate region, we don't need
to run \klee. Note that we do not implement this optimization in our evaluation.} That is, the set of all
candidate variables is $\bigcup \blockknowns[\block'][]$ for all $\block' \in \transmitblocks$.

We need to run \klee separately for every pair of candidate region and candidate variable. Thus, without loss of
generality, we assume for the remainder of this discussion that $\region$ is the candidate region and $x$ is the
candidate variable. By definition, there is a unique block in $\region$, which we'll denote as $\block[R]$, that
dominates all other blocks $R$.

Suppose we are analyzing a program function \texttt{f}. For ease of instrumenting the analysis (below), we
assume it has a single terminating block. If it does not, we modify the function such that it does by creating a
new terminating block and redirecting all previous ones to it. To run \klee, we write a wrapper which serves as
\texttt{main()} and calls \texttt{f} with symbolic arguments. \klee provides an interface for specifying
constraints on the values arguments can take. Setting these constraints correctly is important. For example, if
\texttt{b} in \texttt{f(int* a, int b)} denotes the length of an array pointed to by \texttt{a}, one should
constrain \texttt{b} to be non-negative. Automating deriving argument semantics for this step would be useful,
but we consider it out of scope for this paper.

``Asking'' \klee if the answer to Question~\ref{q:symex} is ``yes'' or ``no'' is done by using an \assert
statement. The assertion is that the answer to Question~\ref{q:symex} is ``no'': the paths upon which $x$ is not
transmitted (if any) are not realizable. To accomplish this, we instrument the program with flags. Let
\texttt{L} denote the flag variable. For simplicity, assume \texttt{L} is not SSA, i.e., can be assigned more
than once. We add the initialization \texttt{L$\;=\;$0} in the entry block of the function \texttt{f}. We add
the assignment \texttt{L$\;=\;$-1} in $\block[R]$. In every $\block' \in \transmitblocks$, we add the assignment
\texttt{L$\;=\;$1}. In the unique terminating block of the function, we add \assert[L$\;\neq\;$-1]. If \klee
manages to find a counterexample to this assumption, then there is some path through $\region$ such that $x$ is
not transmitted. Thus, the answer to Question~\ref{q:symex} is ``yes''. On the other hand, if \klee deduces that
the assertion is provably true, then every realizable trace either circumvents the region $\region$, or enters
it and necessarily transmits $x$, rendering it known. Thus, $x$ is known throughout $\region$.\footnote{One
might worry that this method does not check whether the former is always true, i.e. whether $\region$ is even
ever entered. We do not need to do so since if $\region$ is never entered, knowledge of $x$ in $\region$ is
vacuous and thus inactionable.}

\begin{figure*}[t]
    \centering
\includegraphics[height=3.2cm]{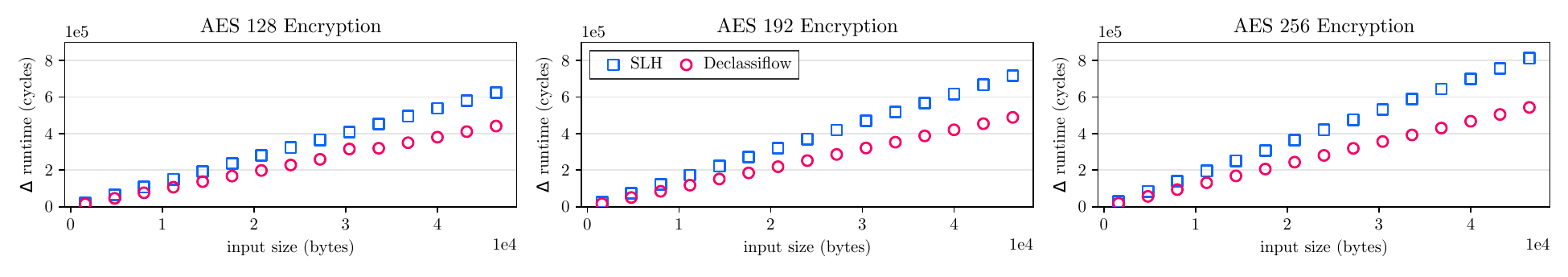}\vspace{-10pt} \caption{The results of running AES encryption
with three different key sizes on inputs of various lengths. We show the raw difference between the runtime (in
cycles) of the functions compiled with SLH and the runtime of the baseline (insecure) code. We show the same
difference for the \toolname protected versions of the functions, i.e., after our analysis is applied. The
overhead reduction from the SLH-enabled code to our protected code grows with the size of the input.}
    \label{fig:eval-aes}
    \vspace{10pt}
\end{figure*}

\subsection{The Protection Pass}
\label{sec:relaxation-pass}

We enforce Property~\ref{prop:frontier} with speculation barriers, which we denote as \specbarr. These barriers
delay the execution of younger instructions until they are non-speculative. \specbarr conceptually implements
\protecc[*]; it indiscriminately applies protection for every variable. That means if the knowledge frontiers
for two variables $x$ and $y$ both include some block $\block$, only one \specbarr needs to be added to $\block$
to enforce the frontier for both $x$ and $y$. Since \specbarr itself is relatively heavyweight, our main tactic
to get speedup will be to determine that frontiers (and therefore \specbarr placement) fall outside of critical
loops. In that case, for sufficiently long loops, the cost of the \specbarr will be amortized.

We now describe the procedure for placing \specbarr for a program function \codevar{f}. We first clone
\codevar{f} to produce \codevarprime{f}. The latter is what we refer to as the \define{protected} version of
\codevar{f}. Without loss of generality, suppose that \texttt{f} (and thus \codevarprime{f}) only makes calls to
one other function \func{g}[][x] which leaks its argument. Assume a protected version of \texttt{g} exists,
denoted \funcprime{g}[][x].

We now describe the procedure to protect the internals of \codevarprime{f}. We compute the joint frontier of
locally transmitted variables, i.e. $F_\text{local} = \bigcup \frontier{v}$ for all \codevar{v} in \texttt{f}.

We next need to get the joint frontier of all variables that are not necessarily locally transmitted, but are
transmitted by calls to other functions; we denote these as $F_\text{func}$. In this example, this would be from
calls to \codevar{g}. We can replace calls to \func{g}[][x] with calls to \funcprime{g}[][x].

The manner in which we enforce protection of \codevarprime{g} depends on whether it is a pseudo transmitter or
not. Suppose it is a pseudo transmitter and thus can be \emph{caller} enforced. For purposes of disambiguation,
let \codevar{x}[i] denote the argument to the \ith call site of \funcprime{g}[][x]. We need to enforce the union
of the frontiers for the leaked \codevar{x}[i]. We compute $F_\text{func} = \bigcup \frontier{x_i}$.\footnote{We
can generalize our discussion to multiple functions being called with multiple arguments by expanding this term.
If there are other functions called by \codevarprime{f} that require caller enforcement, we can add the union
the frontiers of their leaked arguments to this term. If any of the pseudo transmitters called leak multiple
arguments, again, the union of the frontiers of those leaked arguments is added to this term.} If
\codevarprime{g} is \emph{not} a pseudo transmitter, it will be \emph{callee} enforced, and so it will be
internally protected, and we would have $F_\text{func} = \emptyset$.

We now place \specbarr's in \codevarprime{f}. If \codevarprime{f} itself is \emph{not} a pseudo transmitter,
then it must be \emph{callee} enforced. In this case, a \specbarr is placed in every block in $F_\text{local}
\cup F_\text{func}$. However, if \codevarprime{f} is a pseudo transmitter, and thus can be \emph{caller}
enforced, then a \specbarr needs to placed in all those same blocks \emph{except} the entry block of the
function.\footnote{This is because the entry of \codevarprime{f} will be protected at the call site. Note that
if \codevarprime{f} is a top-level function and has no callers, we must place the \specbarr in the entry block.}
Note, if \codevarprime{f} is a pseudo transmitter, then $F_\text{local} \cup F_\text{func}$ is precisely the
entry block of \codevarprime{f}. As a result, \emph{no} \specbarr's will be placed in \codevarprime{f}. This
leads to a useful optimization: if we have a chain of nested calls of pseudo transmitters, we need only a single
\specbarr at the top level to protect the full call chain.

This unidirectional inter-procedural strategy is a fairly greedy method for minimizing the total number of
\specbarr's in the program. Taking a global view of the relationship between functions will undoubtedly lead to
better barrier placement strategies. However, estimating the cost of barriers will most likely rely on
heuristics, and so we leave such a problem to future work.

\section{Evaluation}
\label{sec:eval}

We implement our analysis as an LLVM pass that interoperates with KLEE. We will now evaluate the analysis'
effectiveness in terms of how it can efficiently protect constant-time workloads.

\paragraph{Workloads}
We evaluate 3 constant-time workloads: \circled{1} The AES block-cipher encryption from the \texttt{ctaes}
repository under the Bitcoin Code organization found on Github~\cite{ct-aes}. \circled{2} The Djbsort
constant-time\footnote{It is constant-time with respect to the \emph{values} within the array, \emph{not} the
array size.} integer sorting algorithm~\cite{djbsort,djbsort-paper}. \circled{3} The ChaCha20 stream cipher
encryption from the BearSSL library~\cite{bearssl_chacha20}. For each, we applied \toolname as described in
Sections~\ref{sec:design} and \ref{sec:implementation}.

\paragraph{Baselines}
We compare against each benchmark unmodified and also to each benchmark compiled with Speculative Load Hardening
(SLH)~\cite{slh}. SLH is a Spectre mitigation deployed by LLVM that works by accumulating branch predicate state
and using that state to prevent certain instructions from executing speculatively and/or to prevent certain data
from being forwarded speculatively. The academic community has shown how SLH, when configured to delay the
speculative execution of all transmitters, is sufficient to protect non-speculatively accessed
data~\cite{SSLH,USLH}. We compare to a weaker variant, called ``address SLH'' or aSLH. aSLH was designed to
protect speculatively-accessed data. It does this by delaying the execution of speculative loads, that have
addresses only known at runtime, until they are non-speculative. That is, it considers loads to be instructions
that access (return) sensitive data. Since loads are also transmitters, aSLH can be viewed as implementing a
subset of the mechanisms required to protect non-speculatively accessed data. Thus, the security provided by
\toolname is stronger than aSLH and aSLH's overhead will underestimate the true overhead of SLH in our setting.

All workloads compiled using \toolname maintain the branch predicate state (but do not use it to delay
instructions/data-flows) needed to support SLH. We do this to provide a conservative overhead estimate: SLH
requires this information be maintained across function calls, thus the benchmarks we protect with \toolname can
interoperate with SLH implemented in a calling context, if needed.

\begin{figure}[t]
    \centering
    \includegraphics[height=3.3cm]{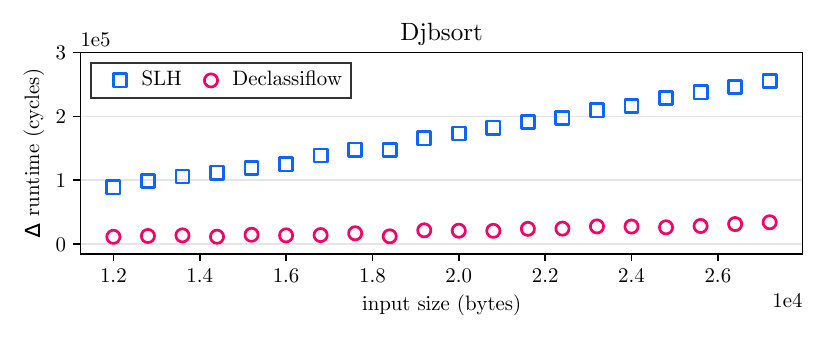}\vspace{-10pt}
\caption{The results of running Djbsort on inputs of various lengths. The
y-axis follows the same convention as in Figure~{\ref{fig:eval-aes}}.}
    \label{fig:eval-djbsort}
\end{figure}

\begin{figure}[t]
    \centering
    \includegraphics[height=3.3cm]{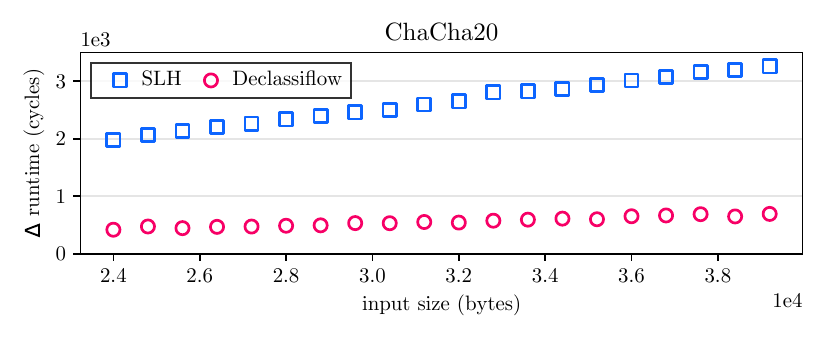}\vspace{-10pt}
\caption{The results of running ChaCha20 on inputs of various lengths.  The y-axis
follows the same convention as in Figure~{\ref{fig:eval-aes}}.}
    \label{fig:eval-chacha20}
\end{figure}

\paragraph{Environment setup}
We run our experiments on an x86-64 Intel Xeon Gold 6148 machine with Ubuntu 20.04, kernel version
5.4.0-146-generic. We compiled our benchmarks with \texttt{-O3} and ensured that all versions of the binary
differ only in their protection mechanisms. The AES source files were compiled with LLVM version 16 while the
Djbsort and ChaCha20 source files were compiled with LLVM version 11.\footnote{These benchmarks needed to be run
with KLEE, the Docker image for which uses LLVM version 11.} We run \klee in a Docker container based on their
official image.\footnote{\texttt{https://hub.docker.com/r/klee/klee/}}

\paragraph{Analysis procedure} To produce protected code, we run the 3 phases of our analysis: \circled{1} The
data-flow analysis is run to determine $\alldfvs$. \circled{2} \klee is used to compute additions to $\alldfvs$.
\circled{3} The protection pass is run, which adds barriers. If \klee is not needed to improve the precision of
the result of phase \circled{1}, phase \circled{2} may be skipped. For the benchmarks which do go through phase
\circled{2}, we provide a static buffer of fixed size that contains symbolic values. We also provide a symbolic
value that represents the buffer's \emph{dynamic} length. We constrain this symbolic value to be as low as 0 and
as high as the length of the static buffer. This static buffer coupled with the length represent the user input
to the functions.

\paragraph{Benchmarking methodology} For constant-time AES encryption,\footnote{We do not benchmark decryption
since the structure (and thus the performance) is nearly identical to encryption.} we benchmark the top-level
functions \texttt{AES128\_encrypt}, \\\texttt{AES192\_encrypt}, and \texttt{AES256\_encrypt}, which use key
sizes of 128, 192, and 256 bits respectively. We benchmark each of these with inputs of various sizes. The
Djbsort and ChaCha20 benchmarks only have one function each that performs the core workload. Thus, we benchmark
those two functions, again with inputs of various sizes. Before every call to the function under test, we
perform a small series of floating point computations which get their input from disparate locations in memory.
The function under test is then only executed if the result is non-zero. This is meant to introduce a branch
that is always taken and will be easily predicted as such. The induced speculation is meant to test the effects
of the \specbarr's we place. We use \texttt{rdtscp} to measure timing. To amortize timer function overhead, we
time the function calls in groups of 8 and then normalize the result. This batch-of-8 timing represents one
``trial''. For every data point, we run 800 trials and discard the first 100 to remove warmup effects of the
cache and TLB. We then report the median of the 700 remaining trials.

\subsection{Main Result}

The full experimental results are presented in Figures~\ref{fig:eval-aes}, \ref{fig:eval-djbsort}, and
\ref{fig:eval-chacha20}. We show the raw difference between the SLH and the \toolname versions of the functions
with respect to the baseline. The geometric means of the relative overheads of SLH for the encryption functions
are 16\%, 16\%, and 17\% for each key size respectively. Meanwhile, the geometric means of the relative
overheads of the \toolname versions of the functions are 12\%, 12\%, and 11\% respectively. For the Djbsort
benchmark, the geometric mean of the relative overhead of SLH is 24\%, while the geometric mean of the relative
overhead of the \toolname version is 3\%. For the ChaCha20 benchmark, the geometric mean of the relative
overhead of SLH is 7\%, while the geometric mean of the relative overhead of the \toolname version is 1\%.

All three benchmarks make heavy use of loops. Thus, the overhead of SLH increases with the size of the input
because the SLH instrumentation is repeatedly encountered. Our analysis tries to prove that the knowledge
frontier is outside of the inner loops (ideally, outside of all loops), hence decreasing the frequency that
protection mechanisms are encountered. With AES, our analysis discovers that the frontier is outside of just the
innermost loops. Thus, while the overhead of the \toolname versions grow with input size, they grow slower than
that of the SLH-protected version. With Djbsort and ChaCha20, our analysis discovers the frontier is completely
outside the main loop. Thus, the overhead of the \toolname versions are low and constant. In the \toolname
version of each of the three benchmarks, only a \emph{single} \specbarr is statically inserted; though for AES
the barrier is placed inside a loop, so it is encountered multiple times dynamically.

\paragraph{Analysis runtime} We now provide details on our analysis' runtime. Before we start, note that our
analysis is \emph{not} part of the typical code-compile-debug workflow programmers use during development.
Instead, we expect it to be run a few times (e.g., once) as a post-processing step to add security on top of
otherwise production-ready code. Thus, these overheads may amortize in practice.

The runtimes of all phases except \circled{2} are given in Table~\ref{tab:phase-times}. The runtime of phase
\circled{2} (running \klee) is over two orders of magnitude higher than the runtimes of the other phases, making
it the bottleneck. Table~\ref{tab:phase-times} also shows how the execution time of \klee scales with respect to
the complexity of the program. In this experiment, we vary the length of the static buffer mentioned previously,
and along with it the constraints on the symbolic length. This essentially changes the \emph{upper bound} on the
size of the arrays that \klee needs to reason about.

To understand better why \klee runtime scales with buffer size, we now report how buffer size impacts \klee's
execution time. Recall from Sections~\ref{sec:symex} and~\ref{sec:symex-pass} that we need to invoke \klee for
every candidate region $\regionraw$ and candidate variable $x$ for which we want to answer
Question~\ref{q:symex}. For Djbsort, there are 72 possible candidate variables and 4 candidate regions; thus, we
need $72\times4=288$ invocations of \klee. For ChaCha20, there are 91 candidate variables and 3 candidate
regions; thus we need $91\times3=273$ invocations of \klee. For both benchmarks, we found that the number of
invocations does not vary with static buffer size. Each invocation evaluates \klee expressions, which generate
what \klee calls ``query constructs.''\footnote{A query construct is a node in the tree created by a \klee
expression.  See \url{https://www.mail-archive.com/klee-dev@imperial.ac.uk/msg02341.html}.} For a given static
buffer size, the number of query constructs generated is the same across all invocations. That said, the number
of query constructs per invocation increases with the static buffer size. This is predictive of overhead;
Figure~\ref{fig:eval-scaling} shows how the time taken for an individual invocation scales closely with the
number of query constructs generated.

\begin{table}[]
\small
\begin{tabular}{l|l|l|l|l|l|l}
& \multicolumn{1}{c|}{\textbf{DFA}} & \multicolumn{4}{c|}{\textbf{Symbolic Execution}} & \multicolumn{1}{c}{\textbf{Protection}} \\
&  & \multicolumn{1}{c}{\footnotesize$N=4$} & \multicolumn{1}{c}{\footnotesize$N=8$} & \multicolumn{1}{c}{\footnotesize$N=16$} & \multicolumn{1}{c|}{\footnotesize$N=32$} &  \\
\toprule
AES & 1s & -- & -- & -- & -- &  $<\,$1s \\
Djbsort & 35s & 37m & 48m & 92m & 266m & $<\,$1s \\
ChaCha20 & 3s & 35m & 38m & 44m & 57m  & $<\,$1s
\end{tabular}
\vspace{20pt}
\caption{\label{tab:phase-times}The runtime of the data-flow analysis (DFA) phase; the protection phase; the
symbolic execution phase (end-to-end runtime for various sizes for the symbolic input buffer). Since AES doesn't
use \klee, the data is omitted. Note the difference in units in the \klee column versus the DFA and relaxation
column. The numbers are the average of 5 runs rounded to the nearest unit.}\vspace{-15pt}
\end{table}

\begin{figure}[t]
\centering
\includegraphics[height=3.5cm]{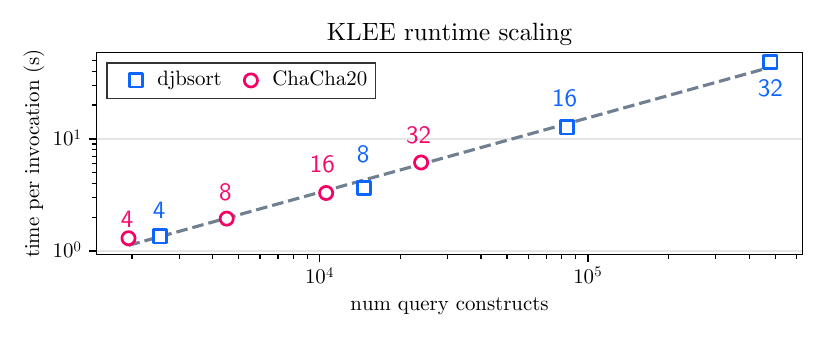}
\vspace{-20pt}
\caption{The time taken per invocation vs. the number of query constructs per invocation. Note, for any
particular symbolic buffer size, every invocation of \klee will have the same number of query constructs. We
plot the median of the time measurements of those invocations. The size of the symbolic buffer associated with
any particular datapoint (chosen to match Table~\ref{tab:phase-times}) is indicated next to it. Note that in
Table~\ref{tab:phase-times}, we are reporting the end-to-end runtime (sum of) for all invocations.}
\label{fig:eval-scaling}
\end{figure}

We spend the remainder of Section~\ref{sec:eval} looking at some of the interesting aspects of each benchmark
and explaining how the analysis responds to those aspects.

\subsection{AES Encryption}
\label{sec:ctaes}

\newsavebox{\aeslisting}
\begin{lrbox}{\aeslisting}
\begin{lstlisting}[basicstyle=\footnotesize\ttfamily,frame=r,linewidth=2.55cm,boxpos=t]
 encrypt($\codevar{y}[1]$) {

  $\;$x$\,\,$=$\,\,$...
  $\;$f(x,$\codevar{y}[1]$)
  $\;$for (...) {
    $\codevar{y}[2]$$\,\,$=$\,\,$$\text{\straightphi}$($\codevar{y}[1]$,$\codevar{y}[3]$)
    g(x,$\codevar{y}[2]$)
    $\codevar{y}[3]$$\,\,$=$\,\,$$\codevar{y}[2]$$\,\,$+$\,\,$1
  $\;$}
  $\;$h(x,$\codevar{y}[3]$)
 }
\end{lstlisting}
\end{lrbox}

\newsavebox{\aesdecllisting}
\begin{lrbox}{\aesdecllisting}
\begin{lstlisting}[basicstyle=\footnotesize\ttfamily,frame=l,linewidth=2.5cm,boxpos=t]
 encrypt($\codevar{y}[1]$) {
  $\;$@SPEC_BARR@
  $\;$x$\,\,$=$\,\,$...
  $\;$!$\color{HotPink}\text{f}'$(x,$\color{HotPink}\codevar{y}[1]$)!
  $\;$for (...) {
    $\codevar{y}[2]$$\,\,$=$\,\,$$\text{\straightphi}$($\codevar{y}[1]$,$\codevar{y}[3]$)
    !$\color{HotPink}\text{g}'$(x,$\color{HotPink}\codevar{y}[2]$)!
    $\codevar{y}[3]$$\,\,$=$\,\,$$\codevar{y}[2]$$\,\,$+$\,\,$1
  $\;$}
  $\;$!$\color{HotPink}\text{h}'$(x,$\color{HotPink}\codevar{y}[3]$)!
 }
\end{lstlisting}
\end{lrbox}

\begin{figure}[t]
    \centering
    \subfloat[Original\label{fig:aes-orig}]{\usebox{\aeslisting}}
    \hfill
    \subfloat[CFG\label{fig:aes-cfg}]{\includegraphics[height=3.52cm,valign=T]{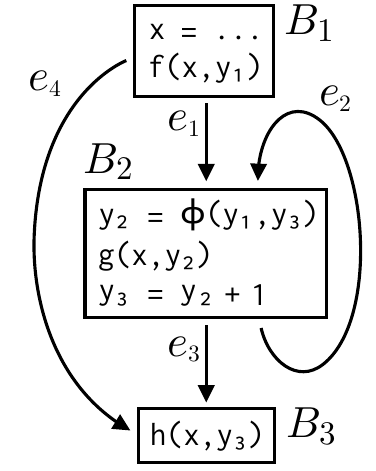}}
    \hfill
\subfloat[\toolname\label{fig:aes-decl}]{\usebox{\aesdecllisting}} \caption{Application of our analysis to
\texttt{AES\_encrypt}. Functions \texttt{f}, \texttt{g}, and \texttt{h} are pseudo transmitters that leak both
their arguments. \pinktt{$\texttt{f}'$}, \pinktt{$\texttt{g}'$} and \pinktt{$\texttt{h}'$} are their \toolname protected
counterparts. NOTE: The figure depicts a \emph{highly simplified} version of the function that still
captures the salient details for our analysis. See Table~\ref{tab:prog-sizes} for the details on the size of the
benchmark.}
    \label{fig:aes-example}
\end{figure}

The functions \texttt{AES128\_encrypt}, \texttt{AES192\_encrypt}, and \\\texttt{AES256\_encrypt} are each given
a plaintext input which can be composed of any number of data blocks.\footnote{Note that in this case,
``blocks'' refers to 16-byte chunks of data, \emph{not} basic blocks.} Each function will call
\texttt{AES\_encrypt} internally for every provided block to encrypt it. \texttt{AES\_encrypt} itself will run a
specified number of rounds of encryption on the provided block. The various steps of a round of AES encryption
are encapsulated in their own functions, which are called by \texttt{AES\_encrypt}.

The analysis' treatment of the function \texttt{AES\_encrypt} is an interesting case study since it can be fully
protected without the need for symbolic execution. Furthermore, it highlights the importance of our
inter-procedural rules. An analogous and \emph{highly simplified} form of the function (that still captures the
salient details for our analysis) and its control-flow graph are depicted in Figures~\ref{fig:aes-orig} and
\ref{fig:aes-cfg}, respectively. (See Table~\ref{tab:prog-sizes} for the details on the size of the benchmark.)
In the figure, \texttt{f}, \texttt{g}, and \texttt{h} represent encapsulations of various operations performed
by AES encryption (e.g. mixing columns or adding in the round key). We point out that the variables \codevar{x}
and \codevar{y}[1] do \emph{not} correspond to secret data in the original code; their contents are the
\emph{addresses} of buffers.

Prior to analyzing \texttt{AES\_encrypt}, \texttt{f}, \texttt{g}, and \texttt{h} will have already been analyzed
and deduced to be pseudo transmitters that leak both of their arguments. From this, the data-flow analysis will
be able to conclude that the frontier for every variable is $\block[1]$. The partial loop expansion mentioned in
Section~\ref{sec:dfa-pass} is crucial to this deduction. Now \texttt{AES\_encrypt} will be protected. There are
no local transmitters to protect. The calls to \texttt{f}, \texttt{g}, and \texttt{h} can be replaced with their
protected counterparts, \pink{\codevarprime{f}}, \pink{\codevarprime{g}}, and \pink{\codevarprime{h}}. 
Being pseudo transmitters, they need to be caller enforced, and so a \bluett{\specbarr} is placed in
$\block[1]$, the frontier of their collective arguments. The \toolname protected version of the code is shown in
Figure~\ref{fig:aes-decl}. Since \texttt{x} is leaked internally and cannot be deduced from the argument
\codevar{y}[1], \texttt{AES\_encrypt} is \emph{not} a pseudo transmitter and thus cannot be caller enforced.

The crucial difference between the SLH and \toolname-protected versions of \texttt{AES\_encrypt} is that under
SLH, the protection of all the variables is present \emph{inside} the loop. Since every loop iteration contains
a branch that can be speculated, the hardening performed by SLH causes repeated delays. Thus, in the graph we
see that the overhead of SLH increases as the size of the input increases. On the other hand, our protected
version of \texttt{AES\_encrypt} has a single \blue{\specbarr}, and within the loop body, we call
\pink{\codevarprime{g}} and \pink{\codevarprime{h}} which do not have hardened loads. Thus, we pay the
\blue{\specbarr} penalty once, but SLH pays a penalty for every load over and over again. This is why our
protected function performs much better than its SLH counterpart.

We note that the higher-level functions which we benchmark call the protected version of \texttt{AES\_encrypt}
(Figure~\ref{fig:aes-decl}) in a loop. Since \texttt{AES\_encrypt} is callee enforced, the number of times the
\blue{\specbarr} is encountered is linear with respect to the size of the input. That is why we see the overhead
of the protected function increase with the size of the input as opposed to remaining fixed.

\subsection{Djbsort}
\label{sec:djbsort}

This benchmark is an interesting case study due to its heavy use of nested loops. The function takes two
parameters, an array of integers \texttt{x} and the length of the array \texttt{N}. All transmitters in Djbsort
arise from accessing \texttt{x} at various offsets. The core concept of the function is depicted in a
\emph{highly simplified} form (that still captures the salient details for our analysis) in
Figure~\ref{fig:djbsort-loop}. (See Table~\ref{tab:prog-sizes} for the details on the size of the benchmark.)
Djbsort cannot be declassified using our data-flow analysis alone. The issue is by the semantics of
\texttt{while} (and similarly \texttt{for}) loops, when the program is compiled, the control-flow graph will
typically include an edge that bypasses the loop body ($\pinkm{\edge[3]}$ from the figure).\footnote{The
compiler can indeed remove this edge if it can be deduced that the loop body will execute at least once.
However, it is unlikely to happen if we have nested loops with complicated conditions, which is what we see in
Djbsort.} Ideally, the frontier of \texttt{x} is $\block[2]$ since \circled{1} \bluett{\transmit[x+i]} in
$\block[2]$ is guaranteed to be encountered once the first branch is crossed, and \circled{2} \texttt{i} will be
known at every point due to being an inductive variable with a known starting value. However, because the
data-flow analysis assumes \pinkm{$\edge[3]$} is traversable, the frontier can be lifted no higher than
$\block[3]$. This problem is addressed by symbolic execution as described in Section~\ref{sec:symex}. It can
deduce that there is no possible execution of Djbsort in which the array \texttt{x} is not accessed at least
once. After the symbolic execution pass, our analysis will correctly report that the frontier of \texttt{x}
should be $\block[2]$.

The conditional \texttt{if$\,$(\greentt{N$\,$<$\,$2})} at the top of the function is crucial for the symbolic
execution tool to make the aforementioned deduction. Once the symbolic execution engine proceeds past this
branch, it is armed with the path constraint $N \geq 2$. This constraint is necessary to deduce that the loop
condition is always satisfied initially, and therefore that \texttt{x} is guaranteed to be leaked.

We note that although symbolic execution is needed to effectively relax Djbsort, it is \emph{not} be able to do
it on its own. Crucially, it is not \texttt{x} that is transmitted, but \texttt{x$\,$+$\,$i}. The data-flow
analysis is needed to deduce that \texttt{i} is always known, and it will do so via partial loop expansion. Only
then can the instrumentation for the symbolic execution pass treat \transmit[x$\,$+$\,$i] as though it leaks
\texttt{x}.\footnote{\transmit[x$\,$+$\,$i] causes \texttt{x$\,$+$\,$i} to be known. Since \texttt{i} is known,
we can exploit backward solvability to deduce that \texttt{x} is known.}

Since \texttt{x} is only guaranteed to be known after it's non-speculatively confirmed that $N\,\geq\,2$, a
\specbarr must be placed after the \texttt{if$\,$($\greentt{N$\,$<$\,$2}$)}. After this point in the program,
SLH can be disabled for the whole function.

\begin{figure}[t]
\centering
\hspace{0.7cm}
\begin{minipage}[H]{3cm}
    \begin{lstlisting}[basicstyle=\footnotesize\ttfamily,frame=r,linewidth=3cm,boxpos=t]
$\;$
$\;$
   if$\,$(^N$\;\greentt{<}\;$2^)$\;${
     return;
   }
   i$\;$=$\;$0
   while$\,$(!i$\;$<$\;$N!)$\;${
     @$\color{HotBlue}{\dots}$$\;\,$=$\;$x[i]@
     i$\;$=$\;$i$\,$+$\,$1
   }
 $\;$
 $\;$
    \end{lstlisting}
\end{minipage}%
\hspace{0.6cm}
\begin{minipage}[H]{3.9cm}
    \includegraphics[height=4cm]{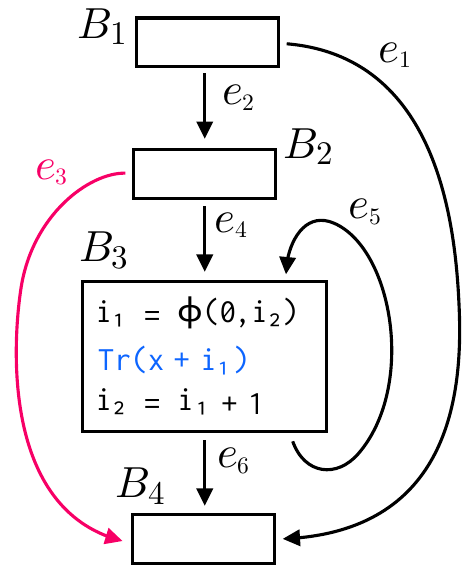}
\end{minipage}
\caption{The core of the Djbsort benchmark from the perspective of our analysis. NOTE: The figure
depicts a \emph{highly simplified} version of the function that still captures the salient details for our
analysis. See Table~\ref{tab:prog-sizes} for the details on the size of the benchmark.}
\label{fig:djbsort-loop}
\end{figure}

\subsection{ChaCha20}
\label{sec:chacha20}

The ChaCha20 benchmark is similar to the Djbsort benchmark in that it involves a series of nested loops and that
all transmitters are due to accesses of various arrays. One can intuit from the source code that if the length
of the provided input is non-zero, then all arrays' addresses are guaranteed to be known. \klee must be used to
prove this in an attempt to safely remove protections for the entirety of the function. One key difference
between ChaCha20 and Djbsort is that in the high-level source code, there is no \texttt{if} statement that
short-circuits the function if a zero-length input is provided. However, compilers will often add such checks in
the form of a loop ``preheader'', and indeed this is what LLVM does for ChaCha20. Thus, from the point of view
of our analysis (which is applied to the LLVM IR generated after compilation), the two benchmarks are quite
similar in form. We thus omit any discussion of the analysis itself. A single \specbarr is placed after the
check for non-zero length but outside any loops, and SLH is disabled for the entire function.

\section{Related Work}

Prior work Blade~\cite{blade} and several recent SLH variants~\cite{USLH, SSLH, SpectreDeclassified} share
\toolname's goal of reducing overhead of speculative execution defenses for constant-time code. Blade protects
only speculatively-accessed data: it statically constructs a data-flow graph from mis-speculated loads (which
can return secrets) to transmitters and infers a minimal placement of protections that cuts off such data-flow.
SSLH~\cite{SSLH} and USLH~\cite{USLH} strengthen SLH to protect non-speculatively accessed data. As discussed in
Section~\ref{sec:eval}, these schemes will incur a higher performance penalty than \toolname while providing
comparable security.

selSLH builds on programming language support for public/secret type information~\cite{FactLanguage}, which it
uses to selectively apply SLH only on loads into public variables. This approach relies on typed variables, with
the type system enforcing that a transmitter's operand is always typed public. However, this property isn't
preserved on compilation: a machine register \texttt{rX} can hold public and secret values in different program
contexts. Therefore, mis-speculation from a context in which \texttt{rX} holds a secret to a context in which
\texttt{rX} is public and is an operand to a transmitter can result in \texttt{rX} leaking~\cite[Section
4.2.2]{deianpldi}. Overall, selSLH, like Blade, doesn't protect secret non-speculative register data. In
contrast, \toolname protects both speculatively-accessed data (read from memory under speculation) and secret
non-speculative register data, and makes no assumptions on the programming language used.

In Shivakumar et al.~\cite{typing-crypto}, the authors propose language primitives for writing cryptographic
code that is protected from Spectre v1. These primitives can be used to implement traditional SLH along with
selSLH, and a type system checks that the program meets a provided security definition. Inserting these
primitives in code is not automatic and requires developer effort. In theory, their work is compatible with ours
as we could use our analysis to relax code protected via their primitives. It could be that more fine-grained
(and thus more beneficial) relaxations would be possible.

Finally, several works~\cite{InSpectre,deianpldi,Spectector,CSFdef} develop static analyses to detect violations
of a formal notion of security against speculative leakage, which is based on the idea that a program's
speculative execution should not leak more than its non-speculative execution. \toolname also leverages this
type of security property, but its goal is to safely relax conservative protections while maintaining security.

\section{Conclusion}

This paper presented \toolname: a static analysis that reduces the amount of protection needed to ``take
speculative execution off the table'' for constant-time programs. The key observation is that as the program's
non-speculative execution makes forward progress, various instructions that reveal their operands over side
channels will \emph{inevitably} execute. Such \emph{inevitably-revealed} operands need not be protected in the
program's speculative execution. This allows one to safely hoist, consolidate or even remove protection
primitives, improving performance.

Longer term, an interesting question will be whether hardware-based schemes such as SPT and software-based
schemes such as \toolname can be combined to further reduce overhead. These two approaches are complementary, in
the sense that a hardware-based scheme can take advantage of fine-grain dynamic information (e.g., the current
path taken by the program) while a software-based scheme can take advantage of global knowledge of the control
data-flow graph. There are also opportunities to improve \toolname as a stand-alone analysis. For example, to
automatically deduce protection placement, understand which program paths are realizable, and improve the
fidelity in which the analysis understands knowledge (e.g., using ideas from QIF~\cite{qif}).

\paragraph{Acknowledgments}
We thank the anonymous reviewers and our (\emph{truly superb}) shepherd for their valuable feedback. This
research was partially funded by NSF grants 1954521, 1942888 and 2154183.

\bibliographystyle{plain}
\balance
\bibliography{refs,chris,architecture,security}

\newpage
\appendix

\begin{table*}[t]
{\small
\begin{tabular}{l|l|l|r|r|r}
\toprule
\textbf{Benchmark} & \textbf{Function} & \textbf{Callees} & \multicolumn{1}{l|}{\textbf{Num Blocks}} &
\multicolumn{1}{l|}{\textbf{Num Edges}} & \multicolumn{1}{l}{\textbf{Num Instrs}} \\
\toprule
\multirow{6}{*}{AES} & AES128\_encrypt & AES\_encrypt & 4 & 7 & 14 \\
& AES\_encrypt & \begin{tabular}[c]{@{}l@{}}LoadBytes, SubBytes,\\ ShiftRows, SaveBytes\end{tabular} & 3 & 6 & 200 \\
& LoadBytes & --- & 3 & 5 & 258 \\
& SubBytes & --- & 7 & 10 & 260 \\
& ShiftRows & --- & 1 & 2 & 177 \\
& SaveBytes & --- & 3 & 5 & 238 \\
\midrule
Djbsort & int32\_sort & --- & 37 & 68 & 205 \\
\midrule
ChaCha20 & br\_chacha20\_ct & --- & 12 & 21 & 409 \\
\bottomrule
\end{tabular}}
\vspace{20pt}
\caption{\label{tab:prog-sizes}The size and call structure of each benchmark. In our analysis, we ensure that
every block has an input and output edge. For singleton-block functions, we thus add two ``dummy edges'',
$(\entrynode, \block)$ and $(\block, \exitnode)$. Notice, Djbsort and ChaCha20 have more complicated
control-flow structure than the AES functions, but the AES functions are nested 3 levels deep, requiring our
inter-function protection rules.}\vspace{-10pt}
\end{table*}

\section{Proofs of Theorems}
\label{app:proofs}

\setcounter{theorem}{0}

\begin{theorem}\label{app:thm:KE_transmitter}Consider an instruction $I$ of the form $\transmit[x]$ in some block
$B$ with output edges $\edge[o]$. For all $\edge[o]$, $x \in \edgeknowns[o]$. \end{theorem}

\begin{proof}
Take any $\edge[o] \in \outedgeset$. Suppose we have a trace $\trace \in \alltraces$ such that $\edge[o] \in
\trace$. By the time $\trace$ traverses $\edge[o]$, it must have already executed $\transmit[x]$. By
Definition~\ref{defn:known}, $x$ is known. So by Definition~\ref{defn:KE}, we must have $x \in \edgeknowns[o]$.
\end{proof}

\begin{theorem}\label{app:thm:KE_fwd_intra}Take any edge $\edge \in \alledges$. Consider a forward solvable
instruction $\inst$ from anywhere in the control-flow graph, and suppose it is of the form $y=f(x_1, \dots,
x_N)$. If $x_i \in \edgeknowns$ for all $x_i$, then $y \in \edgeknowns$.
\end{theorem}

\begin{proof}
Suppose we have a forward solvable instruction of the form $y=f(x_1, \dots, x_N)$, and suppose further that on
some arbitrary edge $\edge$ we have $x_1,\dots,x_N \in \edgeknowns$. By Definition~\ref{defn:KE}, if an
execution traverses $\edge$, then $x_1,\dots,x_N$ are known at that point in time. The attacker can use their
knowledge of these variables along with knowledge of $f$ (due to the program being public) to compute $y$. Thus,
at that point in time, the attacker knows $y$. By Definition~\ref{defn:KE}, $y \in \edgeknowns$.
\end{proof}

\begin{theorem}\label{app:thm:KE_bwd_intra}Take any edge $\edge \in \alledges$. Consider a backward solvable
instruction $\inst$ from anywhere in the control-flow graph, and suppose it is of the form $y=f(x_1, \dots,
x_N)$. For any $j \in \set{1,\dots,N}$, suppose we have all $x_{i\neq j} \in \edgeknowns$ as well as $y \in
\edgeknowns$. Then $x_j \in \edgeknowns$.
\end{theorem}

\begin{proof}
Suppose we have a backward solvable instruction of the form $y=f(x_1, \dots, x_N)$. Suppose further, without
loss of generality, that on some arbitrary edge $\edge$ we have $y, x_1,\dots,x_{N-1} \in \edgeknowns$. By
Definition~\ref{defn:KE}, if an execution traverses $\edge$, then $y, x_1,\dots,x_{N-1}$ are known at that
point in time. The attacker can use their knowledge of these variables along with knowledge of $f$ (due to the
program being public) to compute $x_N$. Thus, at that point in time, the attacker knows $x_N$. By
Definition~\ref{defn:KE}, $x_N \in \edgeknowns$.
\end{proof}

\begin{theorem}\label{app:thm:KE_fwd_inter} Consider a block $\block$. If for some variable $v$ we have $v \in
\bigcap \edgeknowns[i]$ for all realizable $\edge[i]$ in $\inedgeset$, then we have $v \in \edgeknowns[o]$ for
every $\edge[o] \in \outedgeset$.
\end{theorem}

\begin{proof}
Let $V = \bigcap \edgeknowns[i]$ for all realizable $\edge[i]$ in $\inedgeset$. Take any $\edge[o] \in
\outedgeset$. Now take any $\trace \in \alltraces$ such that $\edge[o] \in \trace$. There must be \emph{some}
realizable $\edge' \in \inedgeset$ such that $\edge' \in \trace$. By construction, for any $v \in V$, we must
have $v \in \edgeknowns[\edge'][]$. Note that this is true regardless of what $\edge'$ is. By
Definition~\ref{defn:KE}, when $\trace$ traverses $\edge'$, $v$ is known. By choice of $\trace$, it then
traverses $\edge[o]$. At the time $\edge[o]$ is traversed, $v$ is known. Thus by Definition~\ref{defn:KE}, $v
\in \edgeknowns[o]$.
\end{proof}

\begin{theorem}\label{app:thm:KE_bwd_inter} Consider a block $\block$. If for some variable $v$ we have $v \in
\bigcap \edgeknowns[o]$ for all realizable $\edge[o]$ in $\outedgeset$, and if $v$ is not defined in $\block$,
then we have $v \in \edgeknowns[i]$ for every realizable $\edge[i] \in \inedgeset$.
\end{theorem}

\begin{proof}
Let $V = \bigcap \edgeknowns[o]$ for all realizable $\edge[o]$ in $\outedgeset$. Take any $\edge[i] \in
\inedgeset$ and any $\trace \in \alltraces$ such that $\edge[i] \in \trace$. There must be \emph{some}
realizable $\edge'' \in \outedgeset$ such that $\edge'' \in \trace$. By construction, for any $v \in V$, we must
have $v \in \edgeknowns[\edge''][]$. By Definition~\ref{defn:KE}, when $\trace$ traverses $\edge''$, $v$ is
known. By choice of $\trace$, $\edge[i]$ is guaranteed to be traversed. If $v$ is not defined in $\block$, then
$v$ is guaranteed to be known. Traversing $\edge[i]$ guarantees $v$ is known, and so by
Definition~\ref{defn:KE}, $v \in \edgeknowns[i]$.
\end{proof}

\begin{theorem}\label{app:thm:KE_phi_fwd} Let $\inst[\phi]$ denote a \phifunction of the form $y =\allowbreak
\phi(x_1,\allowbreak\dots,\allowbreak x_N)$ in block $\block$ with $N$ input edges $\edge[1],\dots,\edge[N]$.
The semantics of $\inst[\phi]$ are such that $y = x_i$ if $\edge[i] \in \inedgeset$ is traversed to reach
$\block$. If for all realizable $\edge[i]$ we have $x_i \in \edgeknowns[i]$, then for every output edge
$\edge'$, $y \in \edgeknowns[\edge'][]$.
\end{theorem}

\begin{proof}
Suppose we have a trace $\trace \in \alltraces$ such that for some $\edge' \in \outedgeset$ we have $\edge' \in
\trace$. There must be \emph{some} edge from $\inedgeset$ that is also contained in $\trace$. Without loss of
generality, we'll call this edge $\edge[j]$ and suppose $\edge[j] \in \trace$. By assumption, $x_j \in
\edgeknowns[j]$. By the semantics of $\inst[\phi]$, when traversing $\edge'$, we have $y = x_j$. By
Theorem~\ref{thm:KE_fwd_intra}, $y \in \edgeknowns[j]$ as well. Since we assumed $\edge[j]$ could be any input
edge, $y$ is guaranteed to be known so long as $\edge' \in \trace$. Thus, by Definition~\ref{defn:KE}, $y \in
\edgeknowns[\edge'][]$.
\end{proof}

\begin{theorem}\label{app:thm:KE_phi_bwd} Let $\inst[\phi]$ denote a \phifunction of the form $y =\allowbreak
\phi(x_1,\allowbreak\dots,\allowbreak x_N)$ in a block $\block$ with input edges $\edge[1],\dots,\edge[N]$. The
semantics of $\inst[\phi]$ are the same as in Theorem~\ref{app:thm:KE_phi_fwd}. If for all realizable output
edges $\edge'$ we have $y \in \edgeknowns[\edge'][]$, then for every $\edge[i]$ we have $x_i \in
\edgeknowns[i]$.
\end{theorem}

\begin{proof}
Consider an arbitrary input edge $\edge[i] \in \inedgeset$. Take any output edge $\edge' \in \outedgeset$ and
suppose $y \in \edgeknowns[\edge'][]$. Suppose further we have a trace $\trace \in \alltraces$ such that
$\edge[i], \edge' \in \trace$. By choice of $\trace$, $y = x_i$. We see that when $\edge[i]$ is traversed, we
are guaranteed to know $y$, and as a result, $x_i$. Thus, by Definition~\ref{defn:KE}, $x \in \edgeknowns[i]$.
\end{proof}

\begin{figure}[t]
\centering
\subfloat[A natural loop.%
\label{fig:natural-loop}]{%
    \includegraphics[width=0.47\linewidth]{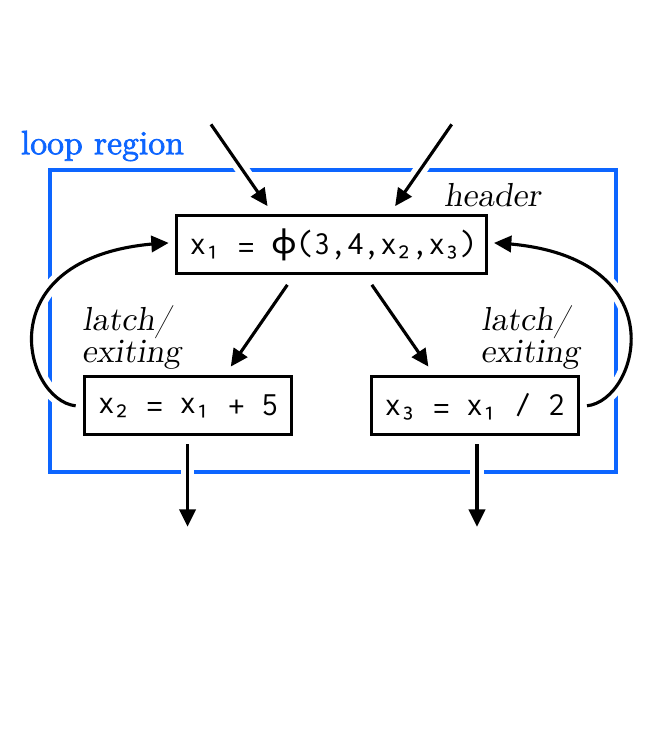}
}
\hfill
\subfloat[The natural loop, ``simplified''.%
\label{fig:simple-loop}]{%
    \includegraphics[width=0.47\linewidth]{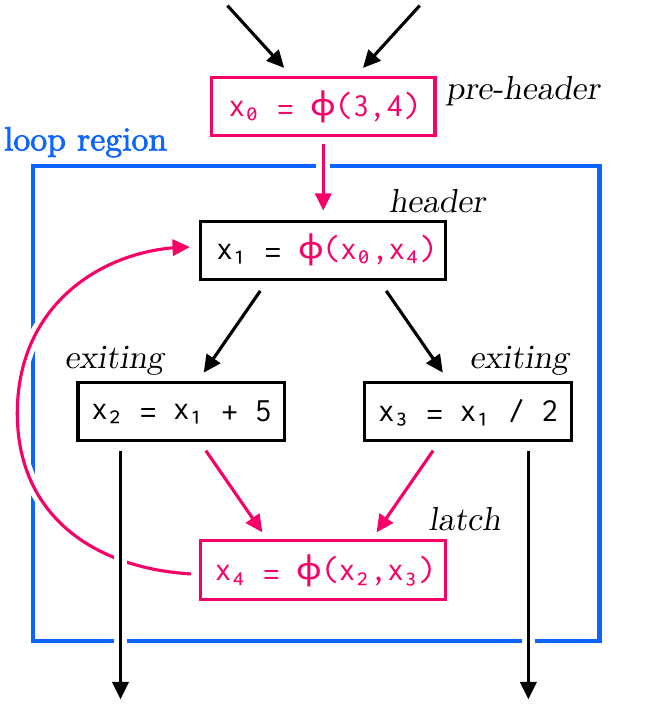}
} \caption{A natural loop (left) is transformed into a simple loop (right) by creating a pre-header node and
merging the two original latch nodes into a single new one. The original \inductivephi is now a two-way
\phifunction. The nodes that were originally exiting blocks remain as such.}
\label{fig:nat-vs-simp-loop}
\end{figure}

\section{Partial Loop Expansion}
\label{app:loop-expansion}

\subsection{The Basic Structure of Loops}
\label{sec:loop-struct}

Generally speaking, a loop is any cycle in the control-flow graph. However, we are primarily concerned with
``natural'' loops \cite{dragon-book}. The \define{natural loop} associated with a back edge
$(\block[j],\block[i])$ is a set of nodes containing $\block[i]$, $\block[j]$, and all nodes from which there
exists a path to $\block[j]$ that does not contain $\block[i]$. The key property of a natural loop is that there
is a single node, the \emph{header}, that dominates all other nodes in the loop. The header node is the
destination of the back edge. If two natural loops have different headers, then either the loops are nested or
are completely disjoint. It is possible for two natural loops to share a header; it's standard to treat such a
case as a single loop with two back edges.

We adopt terminology used by LLVM \cite{llvm-loops}. A \define{latch} is any node that has a back edge
leading to the header of the loop. An \define{exiting} node is one that is inside the loop and has an edge
leading to a node outside the loop. An \define{exit} node is the destination of an exiting node.

We define an \define{\inductivephi} as a \phifunction where at least one of the incoming edges is a back edge.
An \inductivephi selects between two types of values: \circled{1} \define{initial values} --- the initial
versions of a variable that is a constant or another variable from outside the loop; \circled{2}
\define{inductive values} --- the versions of a variable that come from the previous iteration of the loop. An
\define{inductive variable} is any variable that is the output of an \inductivephi. Consider the following loop
in non-SSA form (left) and SSA form (right),

\begin{minipage}[H]{.4\linewidth}
\begin{lstlisting}
    x = 1
    for (...) {
      x++
      print(x)
    }
\end{lstlisting}
\end{minipage}%
\hfill
\begin{minipage}[H]{.55\linewidth}
\begin{lstlisting}
    for (...) {
      $\codevar{x}[1]$ = !$\pink{\straightphi}$(1,$\pink{\codevar{x}[2]}$)!
      $\codevar{x}[2]$ = $\codevar{x}[1]$ + 1
      print($\codevar{x}[2]$)
    }
\end{lstlisting}
\end{minipage}%

The \textcolor{HotPink}{\inductivephi} assigns to the inductive variable \codevar{x}[1]. It selects the initial
value \texttt{1} when the loop is first entered and the inductive value \codevar{x}[2] in every iteration after.

A natural loop may have a header with several predecessors and/or multiple latch nodes. LLVM will attempt to
transform such loops into what we refer to as a \define{simple loop}.\footnotemark\ The first
key property of a simple loop is that the header has only one predecessor, known as the \define{pre-header}. The
second key property is that there is only one back edge into the header \cite{llvm-loops}.
Figure~\ref{fig:nat-vs-simp-loop} shows an example of a natural loop transformed into a simple loop.

In a simple loop, the \inductivephi is always a \emph{two-way} \phifunction. Furthermore, in the first iteration
of the loop, we know the \inductivephi will collapse into a unique initial value while in every subsequent
iteration, the \inductivephi will collapse into a unique inductive value. When dealing with natural loops in
general, things aren't so easy. For example, in Figure~\ref{fig:natural-loop}, we know that the \inductivephi
will return either \texttt{3} or \texttt{4} in the first iteration and either \codevar{x}[2] or \codevar{x}[3]
in subsequent iterations, but we can't conclude anything more than that.

\subsection{Loop Expansion Procedure}

In a nutshell, \define{partial loop expansion} is a transformation which duplicates the loop region. It is a
way to compactly model the dynamic definitions of SSA variables that may occur in loops. The general procedure
is depicted in Figure~\ref{fig:partial-expansion}. Figure~\ref{fig:loop-original} is an abstraction of the
simple loop from Figure~\ref{fig:simple-loop}. The substructure within the loop has been condensed into a single
``super-node'' denoted $\supernode$. Figure~\ref{fig:loop-expanded} shows the partial expansion of the loop.
There are now two copies of the original super-node, denoted $\supernode[1]$ and $\supernode[2]$.
$\supernode[1]$ is identical to $\supernode$ except all \inductivephis have been replaced with assignments of the
\emph{initial} values. $\supernode[2]$ is identical to $\supernode$ except all \inductivephis have been replaced
with assignments of the \emph{inductive} values. Essentially, the path $P\too\supernode[1]\too M\too X_?$ models
the case where a loop is iterated once and the path $P\too\supernode[1]\too\supernode[2]\too M\too X_?$ models
the case where a loop is iterated \emph{at least} twice. Any variable $v$ that is defined in $\supernode$ is
renamed to $\var{v}[][1]$ in $\supernode[1]$ and $\var{v}[][2]$ in $\supernode[2]$. Duplicating the variables
like this allows us to model the dynamic definitions of variables. The new multiple definitions of $v$ need to
be consolidated before reaching any usage of the original variable $v$ by a block after the loop. We can
accomplish this by using a ``merge node'' $M$ that contains the instruction
\texttt{v\;=\;\straightphi(\codevar{v}[][1],\codevar{v}[][2])}.

\begin{figure}[t]
    \centering
    \subfloat[A simple loop.%
    \label{fig:loop-original}]{%
        \includegraphics[height=4.5cm]{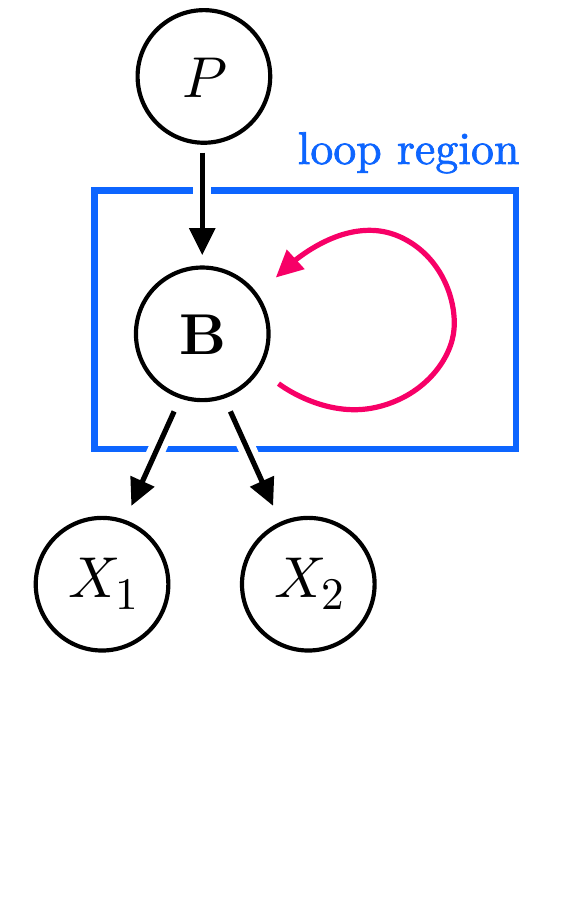}
    }
    \hspace{20pt}
    \subfloat[Its partial expansion.%
    \label{fig:loop-expanded}]{%
        \includegraphics[height=4.5cm]{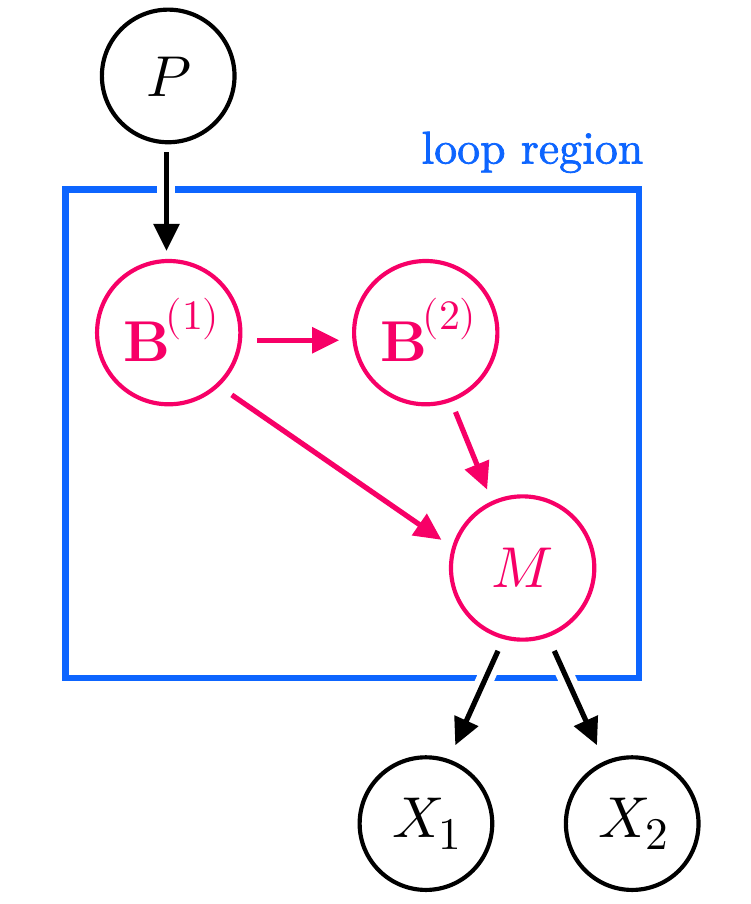}
} \caption{Converting a simple loop into a DAG.\vspace{7pt}}
    \label{fig:partial-expansion}
\end{figure}

\subsection{Example}
\label{sec:loop-example}

\footnotetext{We chose this name because such
loops are the result of the \texttt{LoopSimplify} pass in LLVM; it is not standard LLVM terminology.}

We'll look at the loops in Figure~\ref{app:fig:pita-loops}. They are both identical except in the assignment to
$\codevar{x}[3]$. We use \pinktt{<?>} to represent an unknown value. We'll see that
these two loops have different implications for knowledge and that partial loop expansion is required to treat
them appropriately.

Look at the loop in Figure~\ref{app:fig:pita-loop-1}. If we iterate the loop exactly once, we have that $x_2 =
x_1$, and so \transmit[\codevar{x}[1]] leads to $x_2$ being known. However, if we iterate the loop more than
once, then $x_2$ is actually equal to $x_3$, which is unknown. When an execution traverses $\edge[3]$, whether
$x_2$ is known at that point in time depends on the number of times we iterated the loop. By
Definition~\ref{defn:KE}, we must have $x_2 \notin \edgeknowns[3]$. Now look at
Figure~\ref{app:fig:pita-loop-2}. One can see that the values of \codevar{x}[1], \codevar{x}[2], and
\codevar{x}[3] are all directly connected and thus equivalent in terms of knowledge; so we ought to have $x_2
\in \edgeknowns[3]$. In summary, to achieve both precision \emph{and} soundness, we need $x_2 \notin \dfv[3]$
when analyzing the program in Figure~\ref{app:fig:pita-loop-1} and $x_2 \in \dfv[3]$ when analyzing the program
in Figure~\ref{app:fig:pita-loop-2}.

Figure~\ref{app:fig:fixed-loops} shows the partial expansions of the loops from Figure~\ref{app:fig:pita-loops}.
Looking at the path $\block[1]\too\block[2][1]\too M\too\block[3]$ in both \ref{app:fig:fixed-loop-1} and
\ref{app:fig:fixed-loop-2}, we see there is a direct data-flow from \codevar{x}[1] to \codevar{x}[2] in both
loops. Looking at the path $\block[1]\too\block[2][1]\too\block[2][2]\too M\too\block[3]$, we see there is still
a direct data-flow from \codevar{x}[1] to \codevar{x}[2] in \ref{app:fig:fixed-loop-2}. However, there is no
such data-flow in \ref{app:fig:fixed-loop-1} due to the assignment of \texttt{<?>} to \codevar{x}[3][1].
\codevar{x}[2][2] represents the dynamic definition of \codevar{x}[2] in the second iteration and beyond. Saying
that the data-flow of \codevar{x}[1] doesn't reach \codevar{x}[2][2] is equivalent to saying that after the
first iteration, \codevar{x}[1] and \codevar{x}[2] are no longer connected via data-flow. Running the data-flow
analysis yields $x_2 \notin \dfv[3]$ for Figure~\ref{app:fig:pita-loop-1} and $x_2 \in \dfv[3]$ for
Figure~\ref{app:fig:pita-loop-2}.

This procedure is needed to declassify inductive code, e.g. \texttt{i = 0; while(i < N) \{T(i); i++\}}. This is
crucial for declassifying the benchmarks in Section~\ref{sec:eval}.

\begin{figure}[t]
    \centering
    \subfloat[Potential unsoundness.\label{app:fig:pita-loop-1}]{
        \includegraphics[height=3.2cm]{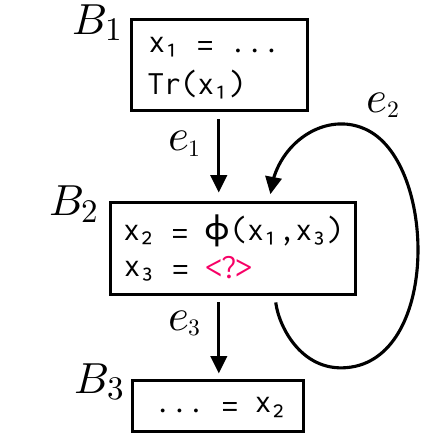}
    }
    \hspace{20pt}
    \subfloat[Potential imprecision.\label{app:fig:pita-loop-2}]{
        \includegraphics[height=3.2cm]{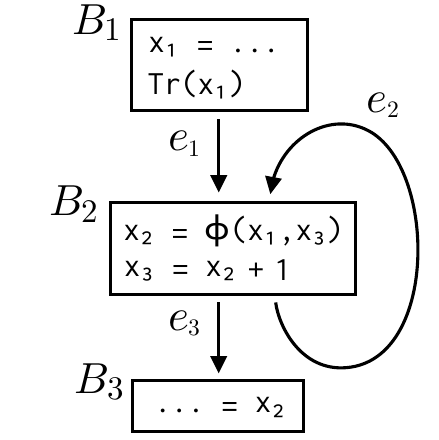}
    }
    \caption{Loops that are problematic for the data-flow analysis unless we use partial loop expansion.
    See Figure~\ref{app:fig:fixed-loops}.\vspace{7pt}}
    \label{app:fig:pita-loops}
\end{figure}

\begin{figure}[b]
    \centering
    \subfloat[Partial expansion of Figure~\ref{app:fig:pita-loop-1}.\label{app:fig:fixed-loop-1}]{
        \includegraphics[height=4.25cm]{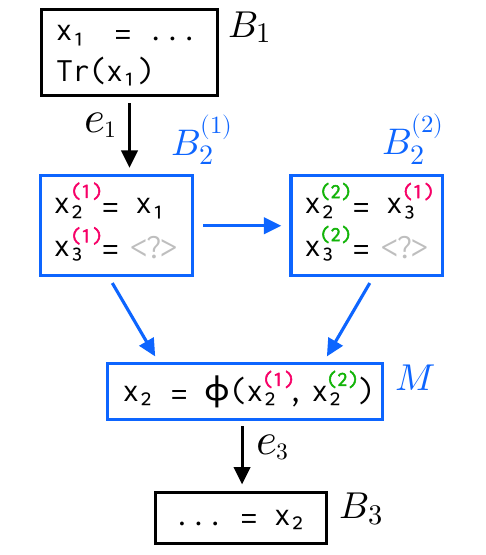}
    }
    \hfill
    \subfloat[Partial expansion of Figure~\ref{app:fig:pita-loop-2}\label{app:fig:fixed-loop-2}]{
        \includegraphics[height=4.25cm]{fixed-loop-2}
    }
    \caption{The partial expansions of the loops from Figure~\ref{app:fig:pita-loops}.}
    \label{app:fig:fixed-loops}
\end{figure}

\end{document}